\documentclass[11pt]{article}
\usepackage{fullpage}
\usepackage{amsmath,amssymb,amsthm}
\usepackage{clrscode}
\usepackage{epsfig}

\usepackage{color}

\usepackage[compact]{titlesec}

\def\shownotes{1}  
\def\full{1} 
\def\spicture{1}

\ifnum \spicture=1
\usepackage{tikz}
\fi

\ifnum\spicture=0
\usepackage[dvips]{graphicx}%
\fi

\newtheorem{fact}{Fact}[section]

\newtheorem{theorem}{Theorem}[section]

\newtheorem{lemma}[theorem]{Lemma}
\newtheorem{proposition}[theorem]{Proposition}
\newtheorem{corollary}[theorem]{Corollary}
\newtheorem{definition}[theorem]{Definition}

\newcommand{\E}{\mathrm{E}}

\newcommand{\uzeta}{\underline{\zeta}}
\newcommand{\ozeta}{\overline{\zeta}}

\newcommand{\rone}{r^{(1)}}
\newcommand{\rtwo}{r^{(2)}}
\newcommand{\hrone}{\hat r^{(1)}}
\newcommand{\hrtwo}{\hat r^{(2)}}
\newcommand{\hmu}{\hat{\mathcal U}}
\newcommand{\hmw}{\hat{\mathcal W}}

\providecommand{\ie}{\emph{i.e.,} }
\providecommand{\eg}{\emph{e.g.,} }

\providecommand{\etc}{\emph{etc.}}      
\providecommand{\mypara}[1]{\smallskip\noindent\emph{#1} }
\providecommand{\myparab}[1]{\smallskip\noindent\textbf{#1} }

\ifnum\shownotes=1
\newcommand{\authnote}[2]{{ $\ll$\textsf{\footnotesize #1 notes: #2}$\gg$}}
\else
\newcommand{\authnote}[2]{}
\fi
\newcommand{\Znote}[1]{{\authnote{Zhenming}{#1}}}

\title{From Black-Scholes to Online Learning: Dynamic Hedging under Adversarial Environments}

\author{Henry Lam \\Boston University \\{\tt khlam@math.bu.edu} \and \quad Zhenming Liu\\ Princeton University \\{\tt zhenming@cs.princeton.edu}}

\begin{document}

\maketitle

\setcounter{page}{0}

\begin{abstract}
We consider a non-stochastic online learning approach to price financial options by modeling the market dynamic as a repeated game between the nature (adversary) and the investor. We demonstrate that such framework yields analogous structure as the Black-Scholes model, the widely popular option pricing model in stochastic finance, for both European and American options with convex payoffs. In the case of non-convex options, we construct approximate pricing algorithms, and demonstrate that their efficiency can be analyzed through the introduction of an artificial probability measure, in parallel to the so-called risk-neutral measure in the finance literature, even though our framework is completely adversarial. Continuous-time convergence results and extensions to incorporate price jumps are also presented.

\end{abstract}

\newpage
\section{Introduction}\label{sec:intro}
In the financial market, an \emph{option} is a contract that gives the holder the right, but not the obligation, to buy or sell an underlying asset or instrument~\cite{hull2009options}. Consider the ``vanilla'' European call option as an example. This contract is controlled by three parameters: the underlying asset $S$, say a stock traded at the New York Stock Exchange, the strike price $K$, and the time to expiration $T$. A holder of a European call option has the right to purchase the stock at time $T$ at the prefixed strike price $K$, regardless of the price of the stock at that time. Suppose that the stock price at time $T$, say $S(T)$ (we shall abuse notation to denote $S(t)$ as the price and $S$ as the stock for convenience), exceeds $K$, the holder will exert his/her right, and \emph{exercise} the option. In the opposite event that $S(T)$ submerges below $K$, a rational holder will initiate no action. Assuming the market is liquid, since the holder who decides to exercise the option can sell the stock immediately at market price, the payoff of this option at time $T$ can be summarized as $\max\{S(T) - K, 0\}$.

In general, any single-instrument European option has a payoff function $g(S(T))$ on an underlying asset $S$, depending on the terms of the contract. An American option, on the other hand, gives the holder the right to exercise the option at \emph{any} time before maturity.


The study of the fair prices of different options has not only been a core area in financial economics, but is also important for market practitioners, given the gigantic volume of options being traded at any trading day \cite{optionvolume}. The area bloomed after the groundbreaking discovery of Black and Scholes~\cite{BS73}, which was later expanded by Merton~\cite{Merton73}. In order to find the fair price of an option, their main idea is to \emph{hedge} against the movement of the underlying asset. Consider for example a European call option. If the purchaser of the option, at time 0, chooses to concurrently sell a portfolio that consists of the underlying asset and cash, in such a way that the portfolio's value at time $T$ is exactly $\max\{S(T)-K,0\}$, then the portfolio must cost the same as the European option; otherwise the option holder can construct an \emph{arbitrage} that generates positive profit at no downside risk, a scenario that is economically prohibited. In this way the fair price of the option is exactly the value of the \emph{replicating} portfolio. 
In the original work in~\cite{BS73}, it was shown that a unique price can be determined using this no-arbitrage principle if the underlying asset's price follows a geometric Brownian motion. Other variety of models may or may not succumb to a unique price. 
See for example~\cite{cont2012financial,stochasticfinance} for comprehensive reviews of the literature.

In this paper we shall consider a non-stochastic online learning framework to price and hedge options: We model option pricing as a repeated game between the trader of the option and the ``adversary" nature that controls the movement of stock price. The motivation is that probabilistic description is not always accurate or easy to model in the financial market (as can be seen by the numerous proposed stochastic models; see, \eg~\cite{bertsimas2001hedging}), which makes it intriguing to study the structure and computational capability if one uses a more conservative, namely adversarial, viewpoint. Broadly speaking, our goal in this paper is to build a systematic framework in the non-stochastic context that allows one to answer fundamental questions that have been around in the stochastic finance community. Two such questions are:

%

\emph{Q1. What is the feature of a non-stochastic model that allows one to relate to the traditional Black-Scholes model, the ``Holy Grail" in option pricing?}

\emph{Q2. Can we design efficient pricing algorithms under the non-stochastic framework and what is the
explicit hedging strategy implied by such algorithms?}



We emphasize that in this paper we attempt to address the above two questions for general options, \ie both European and American, with either convex or non-convex payoffs. We acknowledge that there exist a good number of related works previously, especially in the online learning literature. However, they all appear to be confined to somewhat limited settings. 
For example, the work of Abernethy et~al.~\cite{AFW12} and Abernethy et~al.~\cite{AbernethyBFW13} provide convergence results to Black-Scholes and hedging strategies, but only for convex European options. DeMarzo et~al.~\cite{DKM06} and more recently Gofer and Mansour~\cite{GM11} focus on the adaptive adversary setting and propose regret minimization algorithms; however, their technique is also confined to convex payoff functions. The same holds for the line of work in~\cite{Roorda2005,Ber02a2,ber05b,Kolokoltsov11}, which is also limited to analysis on convex options. Chen et~al.~\cite{Chen10} and Bandi and Bertsimas~\cite{bandi2012tractable}, on the other hand, consider general options (both European and American), but they model the nature as oblivious. In view of these, we aim to offer a non-stochastic model that allows analysis for a much wider range of settings than previously considered. 


\vspace{-.2cm}
\myparab{Our work and contribution.} We consider a model
in which the stock price movement at each step is restricted to a bounded deterministic set. While this model can be simpler than those proposed in some of the previous works, we choose it for two reasons. First, this model, which is the simplest non-stochastic model one can imagine, can be shown to exhibit, in certain cases, similar behavior as the so-called binomial tree model (a discrete version of Black-Scholes). Secondly, this simple model allows us to analyze options that are structurally difficult (such as non-convexity, and American-type), and to propose efficient pricing algorithms. 
%

Indeed, our main message in this paper is that, \emph{for a great generality of options, such a non-stochastic framework offers both structural and algorithmic similarity to the classical Black-Scholes framework, but requires analytical tools that appear to be new in the online learning literature}. This is explained by our two major contributions as follows:

1. We analyze an approximation algorithm to robustly price any options whose payoff function is monotonic and Lipschitz continuous. While the algorithm itself is deterministic (because of the non-stochastic problem nature), its analysis involves the construction of an ``artificial" probability measure, which dramatically tightens the bounds on its error estimates. This ``artificial" measure, resembling the so-called risk-neutral measure~\cite{stochasticfinance} commonly used to derive the Black-Scholes formula, is the first of its kind to analyze any deterministic algorithm. This contribution appears in Section~\ref{sec:non-convex}.

2. We provide algorithmic and structural results to American option pricing, which can be thought as an online game in which the adversary is allowed to withdraw. This withdrawal feature is new in the online learning literature, and our formulation here gives the first natural problem under such setting as well as the first non-trivial analysis. This contribution appears in Section~\ref{sec:genshort}.

Besides these two main contributions, in the Appendix we shall also show several extensions, including: (1) We show that a continuous-time limit of our model converges to the Black-Scholes model for convex options, and non-convergence to Black-Scholes for non-convex options. For the latter case, the limiting price bound is the optimal value of a continuous-time control problem with volatility constraint (Appendix~\ref{sec:control}).  (2) We adapt our model and algorithmic result to incorporate rare jumps in the financial market, which is important because non-smooth price movement is ubiquitous in financial markets, \eg it can model ``volatility smile"~\cite{kou2002jump} (Appendix~\ref{sec:jump}). 

\ifnum\full=0
\myparab{Organization.} We describe our model in Section~\ref{sec:model}. Section~\ref{sec:btree} analyzes the equilibrium of our minimax game and gives convergence results for convex payoff functions. Section~\ref{sec:non-convex} presents algorithm for pricing options with general payoff functions.
Section~\ref{sec:genshort} addresses various generalizations, including a convergence result for options with general payoff, extension of our results to American options, and generalization of our model to incorporate price jumps. The full paper also presents hardness results as well as an alternative way to prove Black-Scholes convergence from a control formulation discussed in Section~\ref{sec:genshort}.
\fi 

\section{Our model}\label{sec:model}
We now describe our model, which has a similar spirit to~\cite{ber05b,AFW12}. Throughout the paper we shall focus on a discrete-time setting, \ie the trader only has the chance to trade at discrete time points. Specifically, consider an option that expires $T$ days from now. We denote time 0 as the time when a transaction of the option occurs, \ie a trader either buys or sells the option. We assume that, before the option expires, the trader has in total $\tau$ time points that allow trade execution. Let these time points be $t = \frac T{\tau}, \frac{2T}{\tau}, ...., T$. Notice that as soon as $\tau$ is decided, the value $T$ is not a parameter in the game (but will reappear when we consider continuous-time limit later on). Throughout our analysis we assume no transaction costs and the market is liquid (\ie the trader can always buy or sell any volume of the asset at the market price at the $\tau$ time points).


We shall model the dynamic of the financial market from time 0 to $T$ as a $\tau$-round two-player game between the trader and nature. Consider an option on the underlying asset $S$, with initial price $S_0$ and the price at the $t$-th round denoted as $S_{t} (=S(tT/\tau))$. For each round $t$, where $1\leq t\leq\tau$, the adversary has complete freedom to choose the return of $S$, given by $R_t\triangleq S_{t}/S_{{t-1}}-1$, within a pre-specified \emph{uncertainty set} $\mathcal U_t$ (we suppress its dependence on $\tau$ for notational convenience).
On the other hand, at the beginning of each round, the trader can choose to long\footnote{We adopt the terminology in finance: to ``long" means to buy, and to ``short" means to sell a product.} $\Delta_t$ dollars' worth of the asset (a negative $\Delta_t$ will imply a short position). At this point we do not impose any capital capacity on the trader, \ie $\Delta_t$ can be as large or small as possible; however, we shall soon see from our analysis that the optimal $\Delta_t$ is bounded and can be explicitly found.

Let us first describe what should be the upper and lower bounds of the option price under no-arbitrage assumption in our model.
To better illustrate our ideas, all of our analysis will assume the risk-free interest rate is $0$. But all our results can be adapted to non-zero interest rates.

\mypara{Upper bound.} Suppose the trader shorts the option at time 0. To hedge his/her position, at each round of the game, the trader decides to buy $\Delta_t$ dollars' worth of the underlying asset. The cumulated payoff to the trader at time $T$ is then given by $\sum_{t = 1}^{\tau}(R_t\Delta_t)-g\left(S_0\prod_{t = 1}^{\tau}(1+R_t)\right)$, plus the option price that he/she gets from selling the option at time 0. Since a rational trader will strive to maximize gain, against an adversary that strives to minimize so, the outcome of this game to the trader will be $\max_{\Delta_t,t\in[\tau]}\min_{R_t \in \mathcal U_t, t \in [\tau]}\sum_{t = 1}^{\tau}(R_t\Delta_t)-g\left(S_0\prod_{t = 1}^{\tau}(1+R_t)\right)$, again plus the initial option price. Now, if the option price is strictly higher than $-\max_{\Delta_t,t\in[\tau]}\min_{R_t \in \mathcal U_t, t \in [\tau]}$ $\sum_{t = 1}^{\tau}(R_t\Delta_t)-g\left(S_0\prod_{t = 1}^{\tau}(1+R_t)\right)$, then shorting the option and carrying out the optimal hedging strategy gives the trader a positive gain at time 0 with no risk, \ie an arbitrage opportunity arises. In other words, the option price at time 0 cannot be higher than
{\small
\begin{equation}\label{eqn:upper}
\min_{\Delta_t, t \in [\tau]}\max_{R_t \in \mathcal U_t, t \in [\tau]}g\left(S_0\prod_{t = 1}^{\tau}(1+R_t)\right)- \sum_{t = 1}^{\tau}(R_t\Delta_t).
\end{equation}
}

\mypara{Lower bound.} The no-arbitrage lower bound can be obtained by reversing the action of the trader from shorting to longing the option at time 0.
Suppose the trader shorts $\Delta_t$ dollars' worth of the underlying asset at the $t$-th round, and strives to maximize $g\left(S_0\prod_{t = 1}^{\tau}(1+R_t)\right)- \sum_{t = 1}^{\tau}(R_t\Delta_t)$. It can be argued similarly that the option price cannot be lower than $
\max_{\Delta_t, t \in [\tau]}\min_{R_t \in \mathcal U_t, t \in [\tau]}g\left(S_0\prod_{t = 1}^{\tau}(1+R_t)\right)- \sum_{t = 1}^{\tau}(R_t\Delta_t)$,
or otherwise arbitrage opportunity arises.



\myparab{Interpretation of bounds.} 
\ifnum\full=1
 Let $u$ be the upper bound from (\ref{eqn:upper}) and $\ell$ be the corresponding lower bound.  
\fi
\ifnum \full=0
Let $u$ and $\ell$ be the upper and lower bounds from our model as described above.
\fi
Their interpretations are the following: (1) When an option price does not fall into $[\ell, u]$, then there exists a trading strategy so that under \emph{any} adversarial scenarios the overall payoff of the strategy is strictly positive (arbitrage exists). Economically speaking, this is a ``wrong" price of the option. (2) When an option price is in $[\ell, u]$, then for any trading strategy, there exists an adversary so that the payoff is non-positive (arbitrage does not exist). This price can be the ``fair" price of the option.

\myparab{Oracle model for payoff functions.} We assume we have oracle access to the payoff function $g(\cdot)$ (which is possibly non-convex). Also, when saying $g(\cdot)$ is Lipschitz continuous, we mean that for any real values $x,y$, we have $|g(y) - g(x)| \leq L|y  - x|$ for some constant $L$.

Note that we leave the choice of the uncertainty set $\mathcal U_t$ here to depend on $t$.
While some of the results later (\eg continuous-time limit) allows for time-dependent $\mathcal U_t$, for simplicity and to highlight the connection with Black-Scholes, much of this paper will focus on $\mathcal U_t$ that do not vary with time, in which case we will merely denote it as $\mathcal U$.

\section{Explicit characterization of equilibrium for the hedging game}\label{sec:btree}
\ifnum\full=0
\fi

In this section we will analyze the equilibrium between the trader and the nature in every single round, which paves the way to our main contributions in the coming sections. 




\ifnum\spicture=1
\tikzstyle{bag} = [text width=2em, text centered]
\tikzstyle{end} = []
\fi

\myparab{A general characterization of equilibrium.}
We first state our characterization results for a single-round game under very general conditions on the uncertainty set and the payoff function.
We will concentrate on the upper bound (\ref{eqn:upper}) in our analysis;
the lower bound can be obtained easily merely by replacing $g$ by $-g$. 

%
%

Suppose $\tau = 1$, and let $\mathcal U$ be a (Borel) measurable set. Our goal is to find the optimal solution of $\min_\Delta\max_{R\in\mathcal U}g(S_0(1+R))-R\Delta$. The following result provides an optimality characterization for any payoff function $g$ and any uncertainty set $\mathcal U$:

\begin{proposition} Let $\tau=1$ and $S_0$ be the initial price. 
Consider a bounded uncertainty set $\mathcal U$ and a continuous payoff function $g$. We have
\begin{equation}\label{eqn:char}
\min_{\Delta}\max_{R \in \mathcal U}g(S_0(1+R)) - R \Delta =  \max_{\substack{P_f \in \mathcal P(\mathcal U),\\ \E_{R\leftarrow P_f}[R] = 0}}\E_{P_f}[g(S_0(1+R))],
\vspace{-.1cm}
\end{equation}
where $\mathcal{P}(\mathcal{U})$ denotes the set of all probability measures $P_f$ that have support on $\mathcal{U}$. The maximization problem in the right hand side above is over all such probability measures that satisfy $\E_{P_f}[R]=0$. \label{prop:dual}
\end{proposition}

From now on we shall call any $P_f$ that satisfies $\E_{P_f}[R]=0$ a risk-neutral probability measure, since it enforces zero expected return (such terminology has been used widely in the no-arbitrage theory in stochastic finance~\cite{stochasticfinance}). Assuming $\mathcal U\neq\{0\}$, the optimal value of \eqref{eqn:char} is finite only when $\mathcal U$ contains at least a point larger than 0 and a point smaller than 0, \eg when $\mathcal U$ is an interval that covers 0; otherwise risk-neutral measure cannot be constructed. 

The proof of Proposition~\ref{prop:dual} relies on a primal-dual argument applied to the following LP, which is equivalent to the minimax problem $\min_\Delta\max_{R\in\mathcal U}g(S_0(1+R))-R\Delta$:
\begin{equation}
\begin{array}{ll}
\mbox{min}&p\\
\text{s.t.}&g(S_0(1+r))-r\Delta\leq p \text{\ \ for any $r\in\mathcal U$}
\end{array}
\end{equation}
where the decision variables are $p$ and $\Delta$. See Appendix~\ref{asec:dual} for the detailed proof. Despite its simplicity, we have not found this proof in previous works; \cite{ber05b, Kolokoltsov11} have obtained the same result as Proposition~\ref{prop:dual} using other geometric arguments.

\myparab{Convex payoff function}
For the special case where the payoff function $g(\cdot)$ is convex and the uncertainty set is an interval, \ie $\mathcal U=[-\uzeta,\ozeta]$ with $\uzeta,\ozeta>0$, we are able to further characterize the solution to the right hand side of (\ref{eqn:char}) as the unique risk-neutral probability distribution with masses concentrated at $-\uzeta$ and $\ozeta$. We can generalize such result to multi-round games via recursion, which turns out to coincide with the binomial tree model (in the case that the uncertainty sets $\mathcal U_t$ do not vary with time; see Appendix~\ref{sec:warmup} for details of binomial tree). We can also prove convergence to the Black-Scholes price in the limit for a class of non-uniform $\mathcal U_t$'s (as time step size shrinks together with the size of $\mathcal U_t$'s suitably; see Appendix~\ref{sec:convexmultiround}). 
In Appendix~\ref{asec:convex} we also discuss the computational complexity of the corresponding hedging strategies.

All the above characterizations however do not extend to non-convex options. In the next section, we will construct and analyze a pricing algorithm for this scenario. In Appendix~\ref{sec:control}, we will show an additional result that the corresponding continuous-time limit is not the Black-Scholes price (\ie not driven by geometric Brownian motion), but is rather given by the solution to a continuous-time control problem with volatility constraint.

\section{Algorithms for non-convex payoffs}\label{sec:non-convex}
This section presents an approximation algorithm for computing the price upper bounds for general payoff functions under the oracle model; the lower bound's algorithm and its analysis is similar and so is omitted here. Throughout this section, we will assume the size of the uncertainty set $\mathcal U=[-\uzeta,\ozeta]$ is uniform across time steps and $\uzeta,\ozeta>0$ are polynomials in $\tau$;
\ifnum \full=1
at the end of this section
\fi
\ifnum \full=0
in the full paper
\fi
we will discuss the non-uniform uncertainty set case. We also assume that the payoff function $g$ is Lipschitz continuous and monotonically non-decreasing, and without loss of generality that $g(0)=0$.


\myparab{Our algorithm.} We use a fairly natural algorithm to approximate the upper bound: we discretize the uncertainty set $\mathcal U$, \ie instead of allowing the adversary to choose an arbitrary value from $\mathcal U$, we only allow it to choose from the discrete set $\hmu \triangleq \{-\uzeta, -\uzeta + \epsilon, -\uzeta + 2 \epsilon, ...., \ozeta\}$,  where $\epsilon$ is a parameter of our algorithm. We call this a \emph{multinomial tree approximation}.

To compute the price upper bound for the multinomial tree, one can use a dynamic program based on the following recursion. Specifically, let $\hat g_t(x)$ be the approximate price upper bound of the option at the $t$-th round. Also, let $b = (\ozeta +\uzeta)/\epsilon$ be the total number of choices an adversary has for each move, and the choices of the return are $r_i = -\uzeta + i \epsilon$ for $i\in[b]$. We compute $\hat g_t(x)$ by finding the optimal solution of the LP:
\begin{equation}\label{eqn:approxrec}
\begin{array}{ll}
\mbox{minimize}&p\\
\text{subject to}&\hat g_{t + 1}(x(1+r_i))-r_i\Delta\leq p \text{\ \ for $i\in[b]$}
\end{array}
\end{equation}
The quantity $\hat g_0(x)$ is then the approximate option price at time 0. Such discretization scheme and backward induction is related to the so-called stochastic mesh method in the area of financial engineering~\cite{broadie2004stochastic}, but the latter uses Monte Carlo and importance sampling to calculate the probability weights at each step instead of carrying out an LP.

\myparab{Analysis.} We shall briefly address some concerns regarding the above scheme. First, since at each step the adversary may choose multiple ways to move the price, it could be worrying that the number of possible states under consideration is exponential in $\tau$. However, so long as $\epsilon$ remains uniform over the rounds, the number of states we need to keep track of at the $t$-th round is
$t(\uzeta+\ozeta)/\epsilon$ (and thus linear in $t$). 
However, note that at the final round, the price $S_{\tau}$ could be any value in an interval of \emph{exponential length} (\ie in the range $[(1-\uzeta)^{\tau}S_0, (1+\ozeta)^{\tau}S_0]$), while our multinomial tree algorithm only ``samples'' polynomial number of points from the function $g(x)$. This implies that on average the distance between any two sampled points is exponential, which also implies that the overall error due to discretization can grow exponentially.

Nevertheless, here is a surprising feature of our algorithm: while a large portion of internal states in the multinomial tree can have additive errors being $\gg \delta S_0$, the aggregate additive error for the value function in the recursion can be shown to have order $\delta S_0$, with the algorithmic running time being polynomial in $\frac{1}{\delta}$ and $\tau$. 
This feature is a consequence of the probabilistic interpretation in our dual formulation, which we will further elaborate in the proof to be presented momentarily.


\ifnum\full=0
This idea resembles the stochastic mesh method in pricing options, most notably American-type, when the underlying asset's price movement is assumed to be stochastic. Assuming a well-defined risk-neutral measure, the stochastic mesh method calculates each $g_t(x)$ through backward induction via Monte Carlo simulation on the return $R_t$, with the use of importance sampling weights. In a sense, our algorithm replaces this importance sampling with an LP in obtaining the weights (the full paper will explain in more detail).

\fi

\ifnum \full=0
Below is the main result in this section (see the proof in the full paper).

\begin{theorem} Let $\delta$ be an arbitrary constant. Consider using the multinomial tree approximation algorithm to find the price upper bound. When $\epsilon = c \delta^2/(L^2\tau^2)$ for some constant $c$, the algorithm gives a $\hat g_0(S_0)$ such that $g_0(S_0) - \delta S_0 \leq \hat g_0(S_0) \leq g_0(S_0)$.
\end{theorem}
\fi

\ifnum \full=1
%
To formalize the above discussion, we start with the first building block regarding the preservation of Lipschitz continuity for the value functions. The lemma below can be proved by using backward induction. See Appendix~\ref{asec:lipschitz} for details.

\begin{lemma}\label{lem:lipschitz}
Suppose the payoff function $g(\cdot)$ is Lipschitz continuous and monotonically non-decreasing, i.e. $g(x)-g(y)\leq L(x-y)$ for $x\geq y$. Then the value function $g_t(\cdot)$ at the $t$-th round is also Lipschitz continuous with the same Lipschitz constant for all $t$.
\end{lemma}


Our second building block here encapsulates the effect of local errors due to discretization. We need to introduce an intermediate quantity $g^m_t(x)$, for each $t\in[\tau]$, defined as the optimal solution of the following linear program:
\begin{equation}\label{eqn:gmt}
\begin{array}{ll}
\mbox{minimize}&p\\
\text{subject to}&g_{t + 1}(x(1+r_i))-r_i\Delta\leq p \text{\ \ for $i\in[b]$}
\end{array}
\end{equation}

The difference between $g_t^m(x)$ and $\hat g_t(x)$ is that the calculation of $g_t^m(x)$ assumes accurate access to the function $g_{t+1}(\cdot)$. Thus, we may view $g^m_t(\cdot)$ as a ``hybrid variable'' that sits in between $g_t(\cdot)$ and $\hat g_t(\cdot)$. The following gives a bound for $g_t(\cdot) - g^m_t(\cdot)$:

\begin{lemma}\label{lem:gtgm}With $g^m_t(x)$ defined in \eqref{eqn:gmt}, we have
\begin{equation}
0 \leq g_t(x) - g^m_t(x) \leq \sqrt{2Lg_t(x) x\epsilon } + L x\epsilon
\end{equation}
where $\epsilon$ is the discretization parameter in the algorithm.
\end{lemma}

The main device we use in its proof is a convenient dual characterization of the optimal solution obtained through the binding constraints.
From this characterization, Lipschitz continuity is then used to bound the magnitude of the local errors. See Appendix~\ref{asec:lem:gtgm} for details.

%
%

We now move to the final stage of our analysis. The following is our main result.
\begin{theorem}\label{thm:main}The approximation $g_0(S_0) - \hat g_0(S_0)$ satisfies
\begin{equation}
0 \leq g_0(S_0) - \hat g_0(S_0) \leq  C_1 \sqrt{\epsilon}L\tau S_0.
\end{equation}
for some constant $C_1$.
\end{theorem}

The proof of this theorem relies heavily on the ``artificial" risk-neutral probability measures that define the optimal dual solutions of the several different value functions, $g_t(x)$, $g_t^m(x)$ and $\hat g_t(x)$, for each step of backward induction. Although these probability measures have no real-world correspondence, they confer ``artificial" martingale properties on the underlying asset's price movement. Moreover, we are in fact granted with some freedom in choosing the measure to work under, and we will see that the one associated with $g_t^m(x)$ is the most effective in truncating the propagation of error. We will now lay out the arguments precisely. 

\begin{proof}
We first apply a standard telescoping trick (see, \eg~\cite{broadie2004stochastic}) to ``decouple'' the local error from the global error, \ie let $d_t(x) = g_t(x) - \hat g_t(x)$ and we have
$$d_t(x) = (g_t(x) - g^m_t(x)) + (g^m_t(x) - \hat g_t(x)).$$
The first term represents the local error at each induction step, whereas the second term comes from error propagation from the future. The first term $g_t(x)-g^m_t(x)$ is handled by Lemma~\ref{lem:gtgm}. For the second term $g^m_t(x) - \hat g_t(x)$, we recall the dual characterization in Proposition \ref{prop:dual} to write
\begin{equation}
  g^m_t(x) - \hat g_t(x) =\max_{\substack{P \in \mathcal P(\hmu) \\ \E_P[R_{t+1}] = 0}}\E\left[g_{t+1}(x(1+R_{t+1}))\right] - \max_{\substack{P \in \mathcal P(\hmu) \\ \E_{P}[R_{t+1}] = 0}}\E\left[\hat g_{t+1}(x(1+R_{t+1}))\right]\notag
 \end{equation}
Notice that $R_{t+1}$ in both expectations above share the same support $\hmu$, the discretized uncertainty set.
Let us write $\E_{x,t}^m$ as the expectation under $P_{x,t}^m: = \arg \max\{\E\left[g_{t+1}(x(1+R_{t+1}))\right]\mid P \in \mathcal P(\hmu), \E_P[R_{t+1}] = 0 \}$.
The above equation becomes
 \begin{eqnarray}
 && \E_{x,t}^m\left[g_{t+1}(x(1+R_{t+1}))\right] - \max_{\substack{P \in \mathcal P(\hmu) \\ \E_{P}[R_{t+1}] = 0}}\left[\hat g_{t+1}(x(1+R_{t+1}))\right] \notag\\
 &\leq & \E_{x,t}^m\left[g_{t+1}(x(1+R_{t+1}))-\hat g_{t+1}(x(1+R_{t+1}))\right]{}\notag\\
 &&{}\text{\ \ (since $P_{x,t}^m$ is feasible for the maximization in the second term above)} \notag\\
 &= & \E_{x,t}^m[d_{t + 1}(x(1+R_{t+1}))] \notag
\end{eqnarray}
Hence, $d_t(x)$ can be bounded recursively under the probability measure $P_{x,t}^m$:
\begin{equation}\label{eqn:recurdt}
d_t(x) \leq \sqrt{2L g_t(x) x\epsilon} + Lx \epsilon + \E^m_{x, t + 1}[d_{t + 1}(x(1+R_{t+1}))].
\end{equation}

Now, let us define a probability measure $P_{S_0}^m$ on the process $\{S_t\}_{t\in[\tau]}$, where $S_t = S_0 \prod_{i = 1}^t(1+R_i)$ is the price of the underlying asset at the $t$-th round. This measure $P_{S_0}^m$ is defined by the stepwise transition probability $P_{x,t}^m$ on each $R_t$, $t\in[\tau]$. From \eqref{eqn:recurdt}, we can expand $d_0(S_0)$ recursively:
\begin{align}
d_0(S_0) & \leq \sqrt{2Lg_t(S_0) S_0\epsilon} +  L S_0\epsilon  + \E_{S_0, 1}^m[d_1(S_1)] \notag \\
 & \leq \sqrt{2Lg_t(S_0)  S_0 \epsilon } + LS_0\epsilon + \E_{S_0, 1}^m[\sqrt{2L g_1(S_1) S_1\epsilon} + L S_1\epsilon ] + \E_{S_0, 1}^m \E_{S_1, 2}^m[d_2(S_2)] \notag \\
 & \quad \quad \mbox{(by expanding $d_1(S_1)$ by (\ref{eqn:recurdt}) again)} \notag \\
&\vdots \notag\\
& \leq \E^m_{S_0}\left[\sum_{t = 1}^{\tau }\left\{\sqrt{2Lg_t(S_t) S_t \epsilon } + LS_t\epsilon \right\}\right] \label{eqn:lastline}.
\end{align}

Next, observe that $\{S_t\}_{t \leq \tau}$ is a martingale under $P_{S_0}^m$ and the filtration $\mathcal F_t=\sigma(R_1,\ldots,R_t)$, since $\E^m_{S_0}[R_{t+1}|\mathcal F_t]=\E_{t,S_t}^m[R_{t+1}]=0$. We leverage this fact to bound both terms in (\ref{eqn:lastline}). First, notice that
$$\E^m_{S_0}[\sum_{t = 1}^{\tau}L \epsilon S_t] = L \epsilon \sum_{t = 1}^{\tau}\E^m_{S_0}[S_t] = L\tau\epsilon S_0$$
For the other term in Eq. \eqref{eqn:lastline}, note that by Lemma~\ref{lem:lipschitz} we have $g_t(x)\leq Lx$ for any $t$ and $x$.
%
Hence we can wrap up the analysis for (\ref{eqn:lastline}):
\begin{eqnarray*}
 \E^m_{S_0}\left[\sum_{t = 1}^{\tau}\left\{\sqrt{2\epsilon L g_t(S_t)S_t } + L\epsilon S_t\right\}\right]  &\leq & \sum_{t = 1}^{\tau}\E^m_{S_0}[\sqrt{2\epsilon L^2 S^2_t}]  + L\tau\epsilon S_0\\
& = & \sqrt{2 \epsilon}L\sum_{t = 1}^{\tau 1}\E^m_{S_0}[S_t] +L\tau\epsilon S\\
& = & \sqrt{2\epsilon}L\tau S_0 + L\tau\epsilon S_0 \leq  C_1 \sqrt{\epsilon}L\tau S_0,
\end{eqnarray*}
for some constant $C_1$.
This completes the proof of Theorem~\ref{thm:main}.
\end{proof}
Thus we have the following corollary:
\begin{corollary} Let $\delta$ be an arbitrary constant. Consider using the multinomial tree approximation algorithm to find the price upper bound. When $\epsilon = c \delta^2/(L^2\tau^2)$ for some constant $c$, the algorithm gives a $\hat g_0(S_0)$ such that $g_0(S_0) - \delta S_0 \leq \hat g_0(S_0) \leq g_0(S_0)$.
\end{corollary}

\myparab{Tightness of the performance.} We shall show in Appendix~\ref{sec:adddep} that as long as $\epsilon$ is a polynomial in $\tau$, an additive error term that is linear in $L$ and $S_0$ will be inevitable under the oracle model. This means that the running time of our algorithm necessarily depends on $L$ and the additive dependency on $S$ is essentially tight.

\myparab{Non-uniform uncertainty set.} When the uncertainty sets are non-uniform, we can still use the multinomial tree algorithm to find the approximate solution so long as the largest uncertainty set is still polynomial in $\tau$. We remark, though, that the parameter $\epsilon$ has to remain unchanged even if the uncertainty sets change over the time, \ie our $\epsilon$ is still set to $\epsilon = c \delta^2/(L^2\tau^2)$.

\fi 

\section{American options}\label{sec:genshort}
We generalize our results to pricing American options, including dual characterization using risk-neutral measure and algorithmic results similar to European options.



As before, let us first consider the single-round game. The upper bound of an American option can be expressed as:\\
\begin{equation}
\min_{\Delta}\max_{\substack{\theta\in \{0, 1\}\\R \in \mathcal U}}\left((g(S_0(1+R))-R\Delta)(1-\theta) + g(S_0)\theta\right)\label{eqn:recuram01}
 \end{equation}
 where $\theta$ is the decision made by the adversary in exercising the option prematurely: $\theta=1$ if early exercise is prompted, otherwise $\theta=0$. We remark that it is the adversary, not the trader, to have the right to exercise early in upper bound evaluation. It is because the upper bound comes from a short-option argument (see Section~\ref{sec:model}) that endows nature as the holder of the option and hence the early exercise right. For lower bound evaluation, the bound is $
\max_{\Delta,\theta\in\{0,1\}}\min_{R \in \mathcal U}\left((g(S_0(1+R))-R\Delta)(1-\theta) + g(S_0)\theta\right)$.
Now let us focus on the upper bound \eqref{eqn:recuram01}, and our first result is that part of the $\max$ and the $\min$ there can be interchanged:
\begin{lemma}\label{lem:interchange}
The optimization \eqref{eqn:recuram01} can be written as
\begin{eqnarray}
& & \min_{\Delta}\max_{\theta\in \{0, 1\}}\max_{R \in \mathcal U}\left((g(S_0(1+R)) - \Delta R)(1-\theta) + g(S_0)\theta\right) \notag\\
& = & \max_{\theta \in \{0, 1\}}\min_{\Delta}\max_{R \in \mathcal U}\left((g(S_0(1+R)) - \Delta R)(1-\theta) + g(S_0)\theta\right)\notag\\
&=&\max\left\{\max_{\substack{P \in \mathcal P(\mathcal U)\\\E[R] = 0}}\E[g(S_0(1+R))],g(S_0)\right\}. \label{characterization American}
\end{eqnarray}
\end{lemma}

\begin{proof}
To simplify notation, let us write $\Phi(\Delta,\theta)\triangleq \max_{R \in \mathcal U}((g(S_0(1+R))-\Delta R)(1-\theta) + g(S_0)\theta)$. Thus, we need to show $\min_{\Delta}\max_{\theta}\Phi(\Delta,\theta) =\max_{\theta}\min_{\Delta}\Phi(\Delta,\theta)$. The direction $\min_{\Delta}\max_{\theta}\Phi(\Delta,\theta) \geq \max_{\theta}\min_{\Delta}\Phi(\Delta,\theta)$ is straightforward. To show that $\min_{\Delta}\max_{\theta}\Phi(\Delta,\theta) \leq \max_{\theta}\min_{\Delta}\Phi(\Delta,\theta)$, the main observation is that $\Delta$ is only influential if $\theta=0$. More precisely, consider $\max_\theta\min_\Delta\Phi(\Delta,\theta)$. When $\theta = 1$, $\min_{\Delta}\Phi(\Delta,\theta) = g(S_0)$, independent of the choice of $\Delta$; when $\theta = 0$, we have $\min_{\Delta}\Phi(\Delta,\theta) = \max_{\substack{P \in \mathcal P(\mathcal U)\\\E[R] = 0}}\E[g(S_0(1+R))]$ by Proposition~\ref{prop:dual}, and there exists $\Delta^* = \arg \max_{\Delta}g(S_0(1+R)) - \Delta R$. Hence
\begin{equation}
\max_{\theta}\min_{\Delta}\Phi(\Delta,\theta) = \max\left\{\max_{\substack{P \in \mathcal P(\mathcal U)\\\E[R] = 0}}\E[g(S_0(1+R))],g(S_0)\right\}
\end{equation}
Consider putting $\Delta^*$ in $\max_{\theta}\Phi(\Delta,\theta)$. We then have
\begin{equation}
\min_{\Delta}\max_{\theta}\Phi(\Delta,\theta) \leq \max_{\theta}\Phi(\Delta^*,\theta) =  \max\left\{ \max_{\substack{P \in \mathcal P(\mathcal U)\\\E[R] = 0}}\E[g(S_0(1+R))],g(S_0)\right\}.
\end{equation}
Thus, we have $\min_{\Delta}\max_{\theta}\Phi(\Delta,\theta) = \max_{\theta}\min_{\Delta}\Phi(\Delta,\theta)$, which renders the single-round price upper bound \eqref{eqn:recuram01} as depicted in \eqref{characterization American}.
\end{proof}
Analogous results hold for lower bound. The significance of Lemma~\ref{lem:interchange} is that we can characterize the optimal solution of the minimax problem \eqref{eqn:recuram01} in terms of risk-neutral probability much like the case of European options, but with an additional outer maximization in the rightmost expression in \eqref{characterization American} to take into account the withdrawal feature. Consequently, we have the following characterization of multi-round game:
\begin{lemma}Consider a $\tau$-round American option hedging game with convex payoff function $g(\cdot)$ and uniform uncertainty set $\mathcal U = [-\uzeta, \ozeta]$. The option price's upper bound is the same as the upper bound for its European counterpart, i.e. it is optimal to exercise at the maturity time.
\label{lemma:american binomial}
\end{lemma}

The proof of this result, which is written in detail in Appendix~\ref{proof american binomial}, uses an inductive argument to show that the price upper bound for convex payoff is given by $g_0(S_0)$, where $g_t(x)$ follows the recursion
\begin{equation}
g_{t}(x) = \max\left\{\E_{P_f}[g_{t+1}(x(1+R))], g(x)\right\} \label{step recursion}
\end{equation}
with $P_f\in\mathcal{P}(\{-\underline{\zeta},\overline{\zeta}\})$ and $\E_{P_f}[R] = 0$. Moreover, one can show by another induction that $g_t(x)\geq g(x)$ (and strict inequality except $t=\tau$ when $g(x)$ is strictly convex), and hence the conclusion.

Note that the above characterization of American option price's upper bound turns out to recover the $\tau$-round binomial tree model, just like in the European option case. Using known results for the binomial tree model~\cite{amin1994convergence}, there is also a suitable continuous-time limit of the adversarial upper bound of American option price to the corresponding Black-Scholes price.

For the case of concave payoff and lower bound calculation, the characterization \eqref{step recursion} holds and it can be shown that it is always optimal to exercise immediately, i.e. the lower bound price is merely $g(S_0)$. For non-convex payoff in upper bound calculation, or non-concave payoff in lower bound calculation, the above reductions do not hold. One can resort to our multinomial tree algorithm presented in Section~\ref{sec:non-convex} to approximate, for instance the upper bound, using the recursive formula $
\hat g_{t}(x) = \max\left\{\max_{\substack{P \in \mathcal P(\hmu)\\ \E[R] = 0}}\E[\hat g_{t+1}(x(1+R))], \hat g(x)\right\}$
where $\hmu$ is the discretized uncertainty set with step length $\epsilon$.
The performance can be analyzed by using the same techniques presented in Section~\ref{sec:non-convex}, which gives the following corollary:
\begin{corollary} Let $\delta$ be an arbitrary constant. Consider using the multinomial tree approximation algorithm to find the American option's upper bound. When $\epsilon = c \delta^2/(L^2\tau^2)$ for some constant $c$, the algorithm gives an $\hat g_0(S_0)$ such that $g_0(S_0) - \delta S_0 \leq \hat g_0(S_0) \leq g_0(S_0)$.
\end{corollary}

We remark that one arguably unsatisfying feature in our model for the upper bound of American options is that the nature can only exercise the option at discrete rounds, \ie at the times when the trader can execute trade decisions. Appendix~\ref{sec:continuous exercise} discusses how this issue can be addressed.


\newpage
{\small
\bibliographystyle{abbrv}	
\bibliography{option}

\begin{thebibliography}{10}

\bibitem{AbernethyBFW13}
J.~Abernethy, P.~L. Bartlett, R.~M. Frongillo, and A.~Wibisono.
\newblock How to hedge an option against an adversary: Black-scholes pricing is
  minimax optimal.
\newblock In {\em NIPS}, pages 2346--2354, 2013.

\bibitem{AFW12}
J.~Abernethy, R.~M. Frongillo, and A.~Wibisono.
\newblock Minimax option pricing meets black-scholes in the limit.
\newblock In {\em STOC}, pages 1029--1040, 2012.

\bibitem{amin1994convergence}
K.~Amin and A.~Khanna.
\newblock Convergence of american option values from discrete-to
  continuous-time financial models1.
\newblock {\em Mathematical Finance}, 4(4):289--304, 1994.

\bibitem{amin1993jump}
K.~I. Amin.
\newblock {Jump diffusion option valuation in discrete time}.
\newblock {\em Journal of Finance}, pages 1833--1863, 1993.

\bibitem{arora2009computational}
S.~Arora and B.~Barak.
\newblock {\em Computational Complexity: A Modern Approach}.
\newblock Cambridge University Press, 2009.

\bibitem{bandi2012tractable}
C.~Bandi and D.~Bertsimas.
\newblock Tractable stochastic analysis in high dimensions via robust
  optimization.
\newblock {\em Mathematical programming}, 134(1):23--70, 2012.

\bibitem{ber05b}
P.~Bernhard.
\newblock The robust control approach to option pricing and interval models: an
  overview.
\newblock In M.~Breton and H.~Ben-Ameur, editors, {\em Numerical Methods in
  Finance}, pages 91--108. Springer, New York, 2005.

\bibitem{Ber02a2}
P.~Bernhard.
\newblock A robust control approach to option pricing including transaction
  costs.
\newblock {\em Annals of the ISDG}, 7:391--416, 2005.

\bibitem{bertsimas2001hedging}
D.~Bertsimas, L.~Kogan, and A.~W. Lo.
\newblock Hedging derivative securities and incomplete markets: An ϵ-arbitrage
  approach.
\newblock {\em Operations research}, 49(3):372--397, 2001.

\bibitem{billingsley}
P.~Billingsley.
\newblock {\em Convergence of Probability Measures}.
\newblock Wiley Series in Probability and Statistics, 2009.

\bibitem{BS73}
F.~Black and M.~Scholes.
\newblock {The Pricing of Options and Corporate Liabilities}.
\newblock {\em The Journal of Political Economy}, 81(3):637--654, 1973.

\bibitem{broadie2004stochastic}
M.~Broadie and P.~Glasserman.
\newblock A stochastic mesh method for pricing high-dimensional american
  options.
\newblock {\em Journal of Computational Finance}, 7:35--72, 2004.

\bibitem{Chen10}
S.~Chen.
\newblock Robust option pricing : An $\epsilon$-arbitrage approach.
\newblock Master's thesis, MIT, 2010.

\bibitem{CHW11}
Q.~Cheng, J.~E. Hill, and D.~Wan.
\newblock Counting value sets: Algorithm and complexity.
\newblock {\em CoRR}, abs/1111.1224, 2011.

\bibitem{cont2012financial}
R.~Cont and P.~Tankov.
\newblock {\em Financial Modelling with Jump Processes, Second Edition}.
\newblock Chapman \& Hall/Crc Financial Mathematics Series. CRC PressINC, 2012.

\bibitem{CRR79}
J.~C. Cox, S.~A. Ross, and M.~Rubinstein.
\newblock {Option pricing: A simplified approach}.
\newblock {\em Journal of Financial Economics}, 7(3):229--263, Sept. 1979.

\bibitem{DKM06}
P.~DeMarzo, I.~Kremer, and Y.~Mansour.
\newblock {Online trading algorithms and robust option pricing}.
\newblock In {\em Proceedings of the thirty-eighth annual ACM symposium on
  Theory of computing}, STOC '06, pages 477--486. ACM, 2006.

\bibitem{durrett2010probability}
R.~Durrett.
\newblock {\em Probability: Theory and Examples}.
\newblock Cambridge Series in Statistical and Probabilistic Mathematics.
  Cambridge University Press, 2010.

\bibitem{optionvolume}
FIA.
\newblock http://www.futuresindustry.org/volume-.asp.
\newblock {\em Trading Volume Statistics}, 2012.

\bibitem{flemingsoner}
W.~H. Fleming and H.~M. Soner.
\newblock {\em Controlled Markov Processes and Viscosity Solutions}.
\newblock Springer, 2006.

\bibitem{GM11}
E.~Gofer and Y.~Mansour.
\newblock Regret minimization algorithms for pricing lookback options.
\newblock In {\em Proceedings of the 22nd international conference on
  Algorithmic learning theory}, ALT'11, pages 234--248. Springer-Verlag, 2011.

\bibitem{stochasticfinance}
A.~S. Hans~Follmer.
\newblock {\em Stochastic Finance: An Introduction in Discrete Time}.
\newblock Walter de Gruyter, 2011.

\bibitem{hull2009options}
J.~Hull.
\newblock {\em Options, Futures, \& Other Derivatives}.
\newblock Prentice Hall finance series. Prentice Hall, 2009.

\bibitem{Kolokoltsov11}
V.~Kolokoltsov.
\newblock Game theoretic analysis of incomplete markets: emergence of
  probabilities, nonlinear and fractional black-scholes equations.
\newblock {\em CoRR}, abs/1105.3053, 2011.

\bibitem{kou2002jump}
S.~G. Kou.
\newblock A jump-diffusion model for option pricing.
\newblock {\em Management science}, 48(8):1086--1101, 2002.

\bibitem{kushner2001numerical}
H.~Kushner and P.~Dupuis.
\newblock {\em Numerical Methods for Stochastic Control Problems in Continuous
  Time}.
\newblock Applications of Mathematics Series. Springer, 2001.

\bibitem{luenberger1997optimization}
D.~Luenberger.
\newblock {\em Optimization by Vector Space Methods}.
\newblock Series in Decision and Control. Wiley, 1997.

\bibitem{Merton73}
R.~C. Merton.
\newblock {Theory of Rational Option Pricing}.
\newblock {\em Bell Journal of Economics and Management Science},
  4(1):141--183, 1973.

\bibitem{Roorda2005}
B.~Roorda, J.~Engwerda, and J.~M. Schumacher.
\newblock Performance of hedging strategies in interval models.
\newblock {\em Kybernetika}, 41(5):[575]--592, 2005.

\bibitem{serfling}
R.~J. Serfling.
\newblock {\em Approximation Theorems of Mathematical Statistics}.
\newblock John Wily \& Sons, Inc., 1980.

\bibitem{steele}
M.~J. Steele.
\newblock {\em Stochastic calculus and financial applications}.
\newblock Springer, 2001.

\end{thebibliography}
}

\appendix
\section{Binomial tree model}\label{sec:warmup}
This section analyzes the standard binomial tree model and connects its analysis to our adversary model. 

\begin{figure}
\centering
\ifnum\spicture=1
\scalebox{0.8}{
\begin{tikzpicture}[sloped]
  \node (a) at ( 0,0) [bag] {$A$\\$\$10.0$};
  \node (b) at ( 4,-1.5) [bag] {$C$ \$9.0};
  \node (c) at ( 4,1.5) [bag] {$B$ \$11.0};\
  \draw [->] (a) to node [below] {$-10\%$} (b);
  \draw [->] (a) to node [above] {$+10\%$} (c);
\end{tikzpicture}
}
\else
A FIGURE OF A ONE-LEVEL BINOMIAL TREE.
\fi
\caption{A one step binomial tree model. The uncertainty set in this example is $\{-10\%, 10\%\}$. The price of the underlying asset is \$10. At the end of the game, the price can either move to \$11 or \$9. }\label{fig:binomial}
\end{figure}
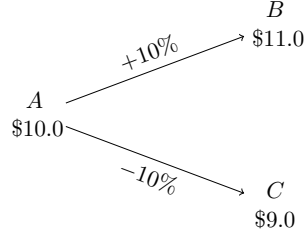

\myparab{Single-round case.} Suppose there is only one round of the game, \ie $\tau = 1$. Here the trader needs only to decide $\Delta$, the amount of the underlying asset $S$ to hold for hedging. In the standard single-round binomial model (see Chapter 12 in~\cite{hull2009options}), the stock price either goes up by a factor of $(1+u)$ or down by a factor of $(1-d)$ (see Figure~\ref{fig:binomial}). In the literature, it is typically assumed that the movement of $S$ is stochastic, i.e. with certain probability $S$ goes up and another probability it goes down. The idea of perfect hedging~\cite{CRR79} is to pick $\Delta$ such that the total payoff at time 1 is constant, regardless of the movement of $S$. In other words, set $\Delta$ that satisfies
$$g(S_0(1+u))-\Delta u=g(S_0(1-d))-\Delta(-d),$$
which gives $\Delta=(g(S_0(1+u))-g(S_0(1-d)))/(u+d)$. Under this hedging strategy $\Delta$, there is no risk for the trader to be compensated \begin{equation}
g(S_0(1+u))-\Delta u=g(S_0(1-d))-\Delta(-d)=\frac{d}{u+d}g(S_0(1+u))+\frac{u}{u+d}g(S_0(1-d)). \label{eqn:binomial price}
\end{equation}
Suppose the option price is different from \eqref{eqn:binomial price}, then an arbitrage opportunity must exist. If the price is higher, the trader shorts the option and longs $\Delta$ dollars' worth of the underlying asset, whereas if the price is lower, the trader longs the option and shorts the same amount of the underlying asset. Both cases lead to risk-free gain to the trader.

Let us now go back to our model described in Section~\ref{sec:model}, with an uncertainty set $\mathcal U=\{u,-d\}$. This is the same as the standard binomial model except that stochasticity of the underlying asset price is now replaced by adversarial movement. The upper bound \eqref{eqn:upper} becomes
\begin{equation}\label{eqn:1rdb}
\min_{\Delta}\max_{R \in \{u,-d\}}g(S_0(1+R)) - R \Delta.
\end{equation}
It is easy to observe that (\ref{eqn:1rdb}) reaches optimum when we set $\Delta$ such that
\begin{equation}\label{eqn:1stageresult}
g(S_0(1+u)) - \Delta u= g(S_0(1-d)) - \Delta (-d),
\end{equation}
or $\Delta=(g(S_0(1+u))-g(S_0(1-d)))/(u+d)$, leading to the same hedging strategy as the standard (stochastic) binomial model. The same argument works for the lower bound and gives rise to the same hedging strategy. We thus have our first basic conclusion: If the uncertainty set in our hedging game is binomial, the upper bound of the option price \emph{matches} the lower bound; moreover, this unique price is the same as the price concluded from the standard (stochastic) binomial model.
%

We also make another observation on the form of our optimal value \ie the equilibrium. Since the optimal hedging amount $\Delta$ for upper and lower bounds are both equal to the standard binomial model, their corresponding optimal values are both given by \eqref{eqn:binomial price}, which can be written as
\begin{equation}
\min_{\Delta}\max_{R \in \{u , -d\}}g(S_0(1+R)) - R \Delta = \max_{\Delta}\min_{R \in \{u , -d\}}g(S_0(1+R)) - R \Delta=\E_{\Sigma}[g(S_0(1+R))]
\end{equation}
where $\Sigma$ assigns probability $d/(u+d)$ to upward movement $u$ and probability $u/(u+d)$ to downward movement $d$. It is easily observed that $\Sigma$ is ``risk-neutral'', \ie $\E_{R \leftarrow \Sigma}[R] = 0$. Hence in this particular case our upper and lower bounds of the option price are both characterized by the same risk-neutral probability measure on the option payoff. We will see that in more general scenarios, the option price bounds can still be characterized by risk-neutral measures, but the measures can be different from each others, and also they can be both different from the measure used in standard binomial pricing. Note that these risk-neutral measures act as analytical artifacts and do not have a real-world correspondence; they will play a key role in our analysis in the rest of this paper.
\\

\myparab{Multi-round case.} Keeping in mind the result above for the single-round case, our price bounds for the multi-round setting can be obtained through straightforward backward induction (dynamic programming). Suppose the game has $\tau$ rounds and each round entails either an up or a down movement for the stock, i.e. $\mathcal U_t=\{u,-d\}$. The upper bound \eqref{eqn:upper} can be written as

\begin{eqnarray}
&&\min_{\Delta_t, t \in [\tau]}\max_{R_t \in \{u,-d\}, t \in [\tau]}g\left(S_0\prod_{t = 1}^{\tau}(1+R_t)\right)- \sum_{t = 1}^{\tau}(R_t\Delta_t)\notag\\
&=&\min_{\Delta_1}\max_{R_1\in\{u,-d\}}\{\cdots\min_{\Delta_{\tau-1}}\max_{R_{\tau-1}\in\{u,-d\}}\{\cdots\min_{\Delta_\tau}\max_{R_\tau\in\{u,-d\}}\{g\left(S_{\tau-1}(1+R_\tau)\right)-R_\tau\Delta_\tau\}\notag\\
& & \quad \quad \quad \quad \quad \quad \quad \quad -R_{\tau-1}\Delta_{\tau-1}\}\cdots-R_1\Delta_1\} \label{eqn:DP}
\end{eqnarray}
The quantity \eqref{eqn:DP} can be solved by iteratively computing $g_\tau(S)=g(S)$ and
$$g_{t-1}(S)=\min_{\Delta_t}\max_{R_t\in\{u,-d\}}g_t(S(1+R_t))-R_t\Delta_t.$$
for $t=\tau,\tau-1,\ldots,1$. By our result above, the solutions to each of these minimax problems are given by $\E_{\Sigma}[g_t(S(1+R))]$. Hence our upper bound again matches the lower bound, and they both match the price according to the standard binomial tree model.

Figure~\ref{fig:2stepbtree} illustrates an example with two rounds. $A$ denotes the state at time 0, $B$ and $C$ at time 1 and so on. We first compute the price of the option at $B$, by analyzing the one-stage game assuming $B$ is the initial time point. The same argument applies to state $C$. We then price the option with initial state $A$ by taking into account the maximal gain the trader can make in the future when the next states are in $B$ or $C$.
%
Observe that there are in total 3 states in this example, instead of $2^2 = 4$, at the end of the second round. In general, the number of states at the final level of a binomial tree grows linearly with the depth of the tree, which makes it a feasible device for option pricing.


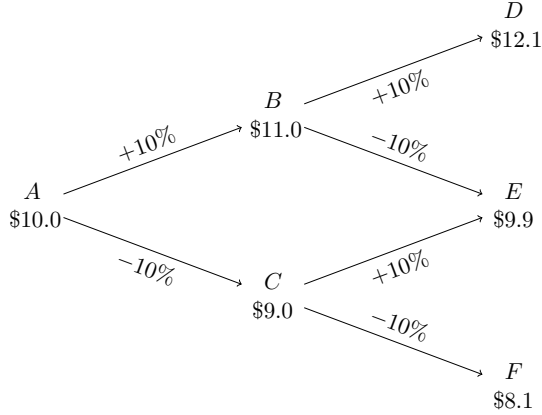
\begin{figure}
\ifnum\spicture=1
\centering
\scalebox{0.8}{
\begin{tikzpicture}[sloped]
  \node (a) at ( 0,0) [bag] {$A$\\$\$10.0$};
  \node (b) at ( 4,-1.5) [bag] {$C$ \$9.0};
  \node (c) at ( 4,1.5) [bag] {$B$ \$11.0};
  \node (d) at ( 8,-3) [bag] {$F$ \$8.1};
  \node (e) at ( 8,0) [bag] {$E$ \$9.9};
  \node (f) at ( 8,3) [bag] {$D$ \$12.1};
  \draw [->] (a) to node [below] {$-10\%$} (b);
  \draw [->] (a) to node [above] {$+10\%$} (c);
  \draw [->] (c) to node [below] {$+10\%$} (f);
  \draw [->] (c) to node [above] {$-10\%$} (e);
  \draw [->] (b) to node [below] {$+10\%$} (e);
  \draw [->] (b) to node [above] {$-10\%$} (d);
\end{tikzpicture}
}
\else
A FIGURE OF TWO-LEVEL BINOMIAL TREE.
\fi
\caption{A two-step binomial tree model.}\label{fig:2stepbtree}
\end{figure}

\myparab{Binomial tree at the limit.} Since our adversary binomial tree model is in effect the same as the standard model, the continuous-time limit converges to the Black-Scholes price under appropriate scaling. The following result is a rephrase of the well-known result in the literature~\cite{hull2009options}:

\begin{proposition}Consider the $\tau$-round European option game. Let the uncertainty set for each round
be $\mathcal U^{\tau} = \{u/\sqrt{\tau},-d/\sqrt{\tau}\}$. Let $g(\cdot)$ be an arbitrary Lipschitz continuous payoff function. The upper and lower bounds of the option with respect to $\mathcal U^{\tau}$ are the same for any $\tau$ and they both converge to the Black-Scholes price as $\tau\to\infty$.
\end{proposition} 
\section{Proof of Proposition~\ref{prop:dual}}\label{asec:dual}
To illustrate the key idea in our analysis, let us start with analyzing a ``discrete" version of the problem, \ie let $\mathcal U = \{r_1, r_2, ..., r_n\}$ be a discrete set on $\mathbb{R}$. We can write $\min_{\Delta}\max_{R \in \mathcal U}g(S_0(1+R)) - R \Delta$ as the following linear program (LP):
{{
\begin{equation}
\begin{array}{ll}
\mbox{min}&p\\
\text{s.t.}&g(S_0(1+r_i))-r_i\Delta\leq p \text{\ \ for $i\in[n]$}
\end{array}
\mbox{ $\equiv$ }
\begin{array}{ll}
\mbox{min}&p\\
\text{s.t.}&p+r_i\Delta\geq g(S_0(1+r_i))\text{\ \ for $i\in[n],$}
\end{array} \label{eqn:primal}
\end{equation}}}
where the decision variables are $p$ and $\Delta$. The formulations in \eqref{eqn:primal} follow
simply by the definition of minimax problem, with the optimal $p$ representing the option price's upper bound. Now we invoke the standard primal-dual theorem for LP on \eqref{eqn:primal} and obtain the following equivalent LP in dual form:

%
{\small
\begin{equation}
\begin{array}{ll}
\mbox{maximize}&\sum_{i\in[n]}w_ig(S_0(1+r_i))\\
\text{subject to}&\sum_{i\in[n]}w_i=1\\
&\sum_{i\in[n]}w_ir_i=0\\
&w_i\geq0\text{\ \ for $i\in[n]$}
\end{array} \label{dual}
\end{equation}
}
Here $w_i$'s are the decision variables. Observe that $\{w_i\}_{i \in [n]}$ can be interpreted as a probability distribution on the uncertainty set $\mathcal U$ since $\sum w_i=1$. Call this distribution $P_f$. This is a risk-neutral probability distribution since the expected return $\E_{P_f}[R]=\sum w_ir_i=0$. Moreover, note that the objective function under this probability interpretation can be rewritten as $\sum_{i\in[n]}w_ig(S(1+r_i)) = \E_{R\leftarrow P_f}[g(S(1+R))]$. We thus have proved (\ref{eqn:char}).

We now generalize the above arguments to general uncertainty set case, with slightly more function space technicalities. For general uncertainty set $\mathcal U$, we can generalize (\ref{eqn:primal}) as:
\begin{equation}
\begin{array}{ll}
\mbox{minimize}&p\\
\text{subject to}&p+r\Delta\geq g(S_0(1+r)) \text{\ \ for $r \in \mathcal U$}
\end{array} \label{eqn:cprimal}
\end{equation}
Next, note that the dual cone of $\mathcal C^+(\mathcal U)$, the set of non-negative continuous functions on $\mathcal U$, is $P^+(\mathcal U)$, the set of positive measures on $\mathcal U$. Since $g$ is assumed to be continuous, the Lagrangian of \eqref{eqn:cprimal} is
$$L(p,\Delta,w)=p+\int_{\mathcal U}(g(S_0(1+r))-r\Delta-p)dw(r)$$
where $w(\cdot)\in\mathcal P^+(\mathcal U)$ (see~\cite{luenberger1997optimization}). The dual function is defined as
\begin{align}
\ell(w)&=\inf_{p,\Delta}\left\{p+\int_{\mathcal U}(g(S_0(1+r))-r\Delta-p)dw(r)\right\}\notag\\
&=\inf_{p,\Delta}\left\{\int_{\mathcal U}g(S_0(1+r))dw(r)+(1-w(\mathcal U))p+\int_{\mathcal U}rdw(r)\Delta\right\} \label{eqn:dual function}
\end{align}
Suppose $w(\cdot)$ does not satisfy either $w(\mathcal U)=1$ or $\int_{\mathcal U}rdw(r)=0$, then one can always find $p$ or $\Delta$ that gives arbitrarily large objective value in \eqref{eqn:dual function}. Hence the dual problem $\max_{w\in\mathcal{P}^+(\mathcal U)}\ell(w)$ can be written as
%
\begin{equation}
\begin{array}{ll}
\mbox{maximize}&\int_{\mathcal U}g(S_0(1+r))dw(r)\\
\text{subject to}&\int_{\mathcal U}rdw(r)=0\\
&w(\mathcal U)=1\\
&w\in\mathcal{P}^+(\mathcal U)
\end{array} \label{eqn:cdual}
\end{equation}

Finally, it is easy to see that the constraint set in \eqref{eqn:cprimal} has non-empty interior (by picking large enough $p$ for example). Hence strong duality holds and the dual optimal value in (\ref{eqn:cdual}) equals the primal counterpart (see \eg Chapter 8 in~\cite{luenberger1997optimization}). By identifying $w$ as a probability measure on $\mathcal U$, we conclude that \eqref{eqn:cdual} is the same as \eqref{eqn:char}. This completes our proof.

\section{Hedging games with convex payoff functions}\label{asec:convex}
This section presents results for games with convext payoffs. We shall start with a corollary of
Proposition~\ref{prop:dual}, regarding one-round games.

Specifically, we will show that it is sufficient to consider risk-neutral probability distributions that have point masses concentrated only on the extremes, namely $-\uzeta$ and $\ozeta$, \ie
\begin{corollary}
When the payoff function $g(\cdot)$ is convex and $\mathcal U=[-\uzeta,\ozeta]$, we have
\begin{equation}
\min_{\Delta}\max_{R \in \mathcal U}g(S_0(1+R)) - R \Delta =  \E_{P_f}[g(S_0(1+R))] \label{eqn:extremes}
\end{equation}
where $P_f\in\mathcal{P}(\{-\underline{\zeta},\overline{\zeta}\})$ and $\E_{P_f}[R] = 0$.
\label{cor:reduction}
\end{corollary}

\begin{proof}
We will show that
\begin{equation}
\min_{\Delta}\max_{R \in \mathcal U}g(S_0(1+R)) - R \Delta =  \max_{\substack{P_f \in \mathcal P(\{-\uzeta,\ozeta\}),\\ \E_{R\leftarrow P_f}[R] = 0}}\E_{P_f}[g(S_0(1+R))] \label{eqn:extremes1}
\end{equation}
where $\mathcal{P}(\{-\underline{\zeta},\overline{\zeta}\})$ is the set of probability distributions that have support only on $-\underline{\zeta}$ and $\overline{\zeta}$. Since $P_f\in\mathcal{P}(\{-\underline{\zeta},\overline{\zeta}\})$ and $\E_{P_f}[R]=0$ uniquely defines $P_f$, the $\max$ operator is redundant. From there we can conclude that (\ref{eqn:extremes}) can be rewritten as
$\E_{P_f}[g(S_0(1+R))]$, where $P_f\in\mathcal{P}(\{-\underline{\zeta},\overline{\zeta}\})$ and $\E_{P_f}[R] = 0$.

We can prove \eqref{eqn:extremes1} by analyzing either the primal program \eqref{eqn:cprimal} in the proof of Proposition \ref{prop:dual} or the characterization \eqref{eqn:char} directly. Let us consider the former as this is more elementary. We argue that, in the case of convex $g$ and $\mathcal U=[-\uzeta,\ozeta]$, the program \eqref{eqn:cprimal} is equivalent to
\begin{equation}
\begin{array}{ll}
\mbox{minimize}&p\\
\text{subject to}&p-\uzeta\Delta\geq g(S(1-\uzeta))\\
&p+\ozeta\Delta\geq g(S(1+\ozeta))
\end{array} \label{eqn:reduced}
\end{equation}
In other words, all other constraints $p+r\Delta\geq g(S(1+r))$ for $r\in(-\uzeta,\ozeta)$ are redundant. To prove this, consider any $-\uzeta<r<\ozeta$. One can write $r=\underline{q}(-\uzeta)+\overline{q}\ozeta$ where $\underline{q}+\overline{q}=1$, $\underline{q},\overline{q}>0$. Suppose the inequalities $p-\uzeta\Delta\geq g(S(1-\uzeta))$ and $p+\ozeta\Delta\geq g(S(1+\ozeta))$ hold. Then
\begin{align*}
p+r\Delta&=\underline{q}(p-\uzeta\Delta)+\overline{q}(p+\ozeta\Delta)\\
&\geq\underline{q}g(S(1-\uzeta))+\overline{q}g(S(1+\ozeta))\\
&\geq g(S(1+r))\text{\ \ (by the convexity of $g$)}
\end{align*}
Therefore, all other constraints are redundant. Now by the same argument as the proof of Proposition \ref{prop:dual} (for discrete uncertainty set), we immediately get \eqref{eqn:extremes1}. The other statement in the corollary follows trivially.
\end{proof}

\subsection{Analysis for the multi-round model}\label{sec:convexmultiround}
For our multi-round game, the trader has the discretion to choose $\tau$ rounds of hedging amount $\{\Delta_t\}_{t\in[\tau]}$ against the nature who controls the $\tau$ rounds of returns $\{R_t\}_{t\in[\tau]}$. We assume a uniform uncertainty set $\mathcal U=[-\uzeta,\ozeta]$ across time.
\ifnum \full=1
The upper bound of the option price is then
\begin{eqnarray}
&&\min_{\Delta_t, t \in [\tau]}\max_{R_t \in \mathcal U, t \in [\tau]}g\left(S_0\prod_{t = 1}^{\tau}(1+R_t)\right)- \sum_{t = 1}^{\tau}(R_t\Delta_t)\notag\\
&=&\min_{\Delta_1}\max_{R_1\in\mathcal U}\{\cdots\min_{\Delta_{\tau-1}}\max_{R_{\tau-1}\in\mathcal U}\{\cdots\min_{\Delta_\tau}\max_{R_\tau\in\mathcal U}\{g\left(S_{\tau-1}(1+R_\tau)\right)-R_\tau\Delta_\tau\}\notag \\
& & \qquad \qquad\qquad \qquad -R_{\tau-1}\Delta_{\tau-1}\}\cdots-R_1\Delta_1\} \label{eqn:recur}
\end{eqnarray}
\fi

\ifnum \full=0
The following lemma is a consequence of the result 
in Proposition~\ref{prop:dual} (see the full version for the analysis)
\footnote{}:
\fi

\ifnum \full=1
The following lemma is a consequence of the result 
in Proposition~\ref{prop:dual}.
\fi

\begin{lemma}\label{lem:discrete}Consider the $\tau$-round hedging game with the same uncertainty set $\mathcal U=[-\uzeta,\ozeta]$ across time and convex payoff function $g(\cdot)$. The upper bound of the option price is $\E_{P_f}[g(S_0\prod_{t=1}^\tau(1+R_t))]$, where $P_f$ is the unique risk-neutral probability distribution on $\{-\uzeta,\ozeta\}$ for all $\{R_t\}_{t\in[\tau]}$, \ie $P_f\in\mathcal(\{-\uzeta,\ozeta\})$ and $\E_{P_f}[R_t]=0$ for all $t\in[\tau]$. 
\end{lemma}

\begin{proof}
The proof is a direct application of dynamic programming, coupled with the preservation of convexity across iterations of the value functions. First, observe the following:

\begin{fact}\label{lem:convex2convex} Let $h(\cdot)$ be an arbitrary convex function. Then $\E_P[h(S(1+R))]$ is convex in $S$, where $P$ is an arbitrary distribution for $R$.
\end{fact}

The statement is immediate by using linearity of expectations and the assumption that $h(\cdot)$ is convex.

Next, we can write \eqref{eqn:recur} as a dynamic program, given by $g_\tau(x)=g(x)$ and $$g_{t-1}(x)=\min_{\Delta_t}\max_{R_t\in\mathcal U}g_t(x(1+R_t))-R_t\Delta_t$$
for $t=\tau,\tau-1,\ldots,1$. The price upper bound is then given by $g_0(S_0)$. We prove by induction that $g_t(\cdot)$ are all convex and $g_{t-1}(x)=\E_{P_f}[g_t(x(1+R_t))]$. The statement is obvious for $g_\tau(\cdot)$. Now, supposing $g_t(\cdot)$ is convex, we have from Corollary~\ref{cor:reduction} that $g_{t-1}(x)=\E_{P_f}[g_t(x(1+R_t))]$, and from Lemma \ref{lem:convex2convex} that $g_{t-1}(\cdot)$ is convex. Hence the induction holds.

\end{proof}

\myparab{Explicit hedging strategy.} When the uncertainty sets are uniform intervals and the payoff function is convex, the optimal hedging strategy is straightforward (and is identical to the binomial model): $\Delta_t = \frac{g(S_{t-1}(1+\ozeta)) - g(S_{t-1}(1-\uzeta))}{\uzeta + \ozeta}$ dollar on $S$ for each round $t$.

\myparab{Non-uniform uncertainty sets.} When the uncertainty sets are non-uniform, say $\mathcal U_t \triangleq [-\uzeta_t, \ozeta_t]$, Corollary~\ref{cor:reduction}, Lemma~\ref{lem:convex2convex} and the form of the hedging strategy all still hold with the natural modification. This means at each round we need only consider the two points $\{-\uzeta_t, \ozeta_t\}$.
However, from a computational point of view, the number of states we need to keep track of in the backward induction could grow exponentially in $\tau$. 
Thus, a naive application of dynamic programming algorithm will not be efficient. In Appendix~\ref{a:hard} we show that exact computation of the option's upper bound with non-uniform uncertainty sets is $\#P$-hard, and in Section~\ref{sec:non-convex} we shall design an approximation algorithm to solve the problem.

\myparab{Convergence.} An immediate implication of Corollary~\ref{cor:reduction} and Lemma~\ref{lem:discrete} is that the upper bound of the option price, when the uncertainty set is $\mathcal U=[-\uzeta,\ozeta]$ across time steps, collides with the binomial tree model that either goes up by $1+\ozeta$ or down by $1-\uzeta$ at each step.
%

It is known that the price from the binomial tree model converges to Black-Scholes~\cite{CRR79}, as the number of time steps increases and the interval length decreases at a rate equal to the square root of the number of time steps. This implies the convergence of our upper bound, with uncertainty sets $\mathcal U_t^\tau=[-\uzeta/\sqrt{\tau},\ozeta/\sqrt{\tau}]$, also to the Black-Scholes price. Here we state a convergence result that is more general: as long as the (possibly non-uniform) collection of uncertainty sets follow a ``bounded quadratic variation" condition, we obtain convergence to the Black-Scholes price for European-type options:

\begin{corollary}\label{cor:limit}Consider the $\tau$-round hedging game with Lipschitz continuous convex payoff function $g(\cdot)$. Let $\left\{\{\mathcal U^{\tau}_t\}_{t \leq \tau}\right\}_{\tau \geq 1}$ be the sequence of uncertainty sets and let $\mathcal U^{\tau}_t = [-\uzeta^{\tau}_t, \ozeta^{\tau}_t]$. If $\lim_{\tau\to\infty}\sum_{t=1}^\tau\uzeta^\tau_t\ozeta^\tau_t=\nu$ for a positive number $\nu$, and $\sup_{t\in[\tau]}\max\{\uzeta^\tau_t,\ozeta^\tau_t\}\to0$, the upper bound of the European option price converges to $\E[g(S_0\exp\{\sqrt{\nu}N(0,1)-\nu/2\})]$, where $N(0,1)$ is standard Gaussian variable, \ie it converges to the option price for a geometric Brownian motion with zero drift.
\end{corollary}

We remark that if $\nu=\sigma^2T$ for a positive constant $\sigma^2$, then the condition $\lim_{\tau\to\infty}\sum_{t=1}^\tau\uzeta^\tau_t\ozeta^\tau_t=\nu$ imitates the quadratic variation of a Brownian motion, and the result recovers the Black-Scholes price. The uniform convergence condition $\sup_{t\in[\tau]}\max\{\uzeta^\tau_t,\ozeta^\tau_t\}\to0$ is necessary; there is no guarantee of Gaussian convergence in the limit if one uncertainty set keeps constant size as $\tau\to\infty$.

\begin{proof}[Proof of Corollary~\ref{cor:limit}]
From Lemma \ref{lem:discrete}, the upper bound of the option price, for any $\tau$, is $\E_{P_f}[g(S_0\prod_{t=1}^\tau(1+R_t^\tau))]$ where $R_t^\tau$ is the $t$-th round return in a $\tau$-round game, and $P_f$ (which also depends on $\tau$) is the unique probability measure that satisfies $\E_{P_f}[R_t^\tau]=0$ and has support $\{-\uzeta^\tau_t, \ozeta^\tau_t\}$ on $R_t^\tau$ for any $t\in[\tau]$. Simple calculation reveals that $P_f$ puts weights $\ozeta^\tau_t/(\uzeta^\tau_t+\ozeta^\tau_t)$ on $-\uzeta^\tau_t$ and $\uzeta^\tau_t/(\uzeta^\tau_t+\ozeta^\tau_t)$ on $\ozeta^\tau_t$. This implies that $\E_{P_f}[(R_t^\tau)^2]=\uzeta\ozeta$.

We shall prove that
\begin{equation}
\log(S_0\prod_{t=1}^\tau(1+R_t^\tau))=\log S_0+\sum_{t=1}^\tau R_t^\tau-\sum_{t=1}^\tau\frac{(R_t^\tau)^2}{2}+\sum_{t=1}^\tau\frac{\xi(R_t^\tau)}{3} \label{eqn:log S}
\end{equation}
where $\xi(R_t)$ satisfies $|\xi(R_t^\tau)|\leq C|R_t^\tau|^3$ for a constant $C$, converges in distribution to $\log S_0+N(0,1)-\nu/2$.

Now consider each term in \eqref{eqn:log S}, and we start with $\sum_{t=1}^\tau R_t^\tau$. Since $\sum_{t=1}^\tau \E_{P_f}[(R_t^\tau)^2]=\sum_{t=1}^\tau\uzeta^\tau_t\ozeta^\tau_t\to\nu$, and $\sum_{t=1}^\tau \E_{P_f}[(R_t^\tau)^2;|R_t^\tau|>\epsilon]$ is eventually zero as $\tau\to\infty$, for any $\epsilon>0$, by Lindeberg-Feller Theorem (p. 114, (4.5) in~\cite{durrett2010probability}), we have $\sum_{t=1}^\tau R_t^\tau\to\sqrt{\nu}N(0,1)$ in distribution.

Next consider the term $\sum_{t=1}^\tau\frac{(R_t^\tau)^2}{2}$. By our condition $\sup_{t\in[\tau]}\max\{\uzeta^\tau_t,\ozeta^\tau_t\}\to0$, since $\sum_{t=1}^\tau \Pr(|R_t^\tau|>\epsilon)$ is eventually zero as $\tau\to\infty$, for any $\epsilon>0$, and also $\sum_{t=1}^\tau \E_{P_f}(R_t^\tau)^4\leq\sum_{t=1}^\tau \E_{P_f}(R_t^\tau)^2\cdot\sup_{t\in[\tau]}(R_t^\tau)^2\to0$, the Weak Law for Triangular Arrays hold (p. 40, (5.5) in~\cite{durrett2010probability}), and $\sum_{t=1}^\tau(R_t^\tau)^2\to\sum_{t=1}^\tau \E_{P_f}[(R_t^\tau)^2]=\nu$ in probability.

For the last term, we have $|\sum_{t=1}^\tau\xi(R_t^\tau)|\leq C\sum_{t=1}^\tau(R_t^\tau)^2\sup_{t\in[\tau]}|R_t^\tau|\to0$ in probability. Combining all these terms, by Slutsky's Theorem (see \eg p. 19 in~\cite{serfling}), we conclude that $\log(S_0\prod_{t=1}^\tau(1+R_t^\tau))\to\sqrt{\nu}N(0,1)-\nu/2$ in distribution.

Lastly, we will conclude our result by checking a uniform integrability condition (see \eg p. 14 in~\cite{serfling}). First, since $g$ is continuous, the Continuous Mapping Theorem~\cite{billingsley} stipulates that $g(S_0\prod_{t=1}^\tau(1+R_t^\tau))$ converges to $g(\exp\{\sqrt{\nu}N(0,1)-\nu/2\})$ in distribution. We shall show that $\sup_{\tau}E_{P_f}[g(S_0\prod_{t=1}^\tau(1+R_t^\tau))^2]<\infty$. This will imply that $g(S_0\prod_{t=1}^\tau(1+R_t^\tau))$ is uniformly integrable, which will then conclude the convergence in $L_1$ of $g(S_0\prod_{t=1}^\tau(1+R_t^\tau))$ into $g(S_0\exp\{\sqrt{\nu}N(0,1)-\nu/2\})$ and conclude our result. To this end, note that
{\small
\begin{eqnarray*}
& &\E_{P_f}\left[g(S_0\prod_{t=1}^\tau(1+R_t^\tau))^2\right]\\
&\leq&C_1\E_{P_f}\left|S_0\prod_{t=1}^\tau(1+R_t^\tau)\right|^2+C_2{}\\
&&{}\text{\ \ (for some constants $C_1,C_2>0$, since $g$ is assumed to be Lipschitz continuous)}\\
&=&C_1S_0\prod_{t=1}^\tau E_{P_f}(1+R_t^\tau)^2+C_2\text{\ \ (by independence of $R_t^\tau$)}\\
&=&C_1S_0\prod_{t=1}^\tau(\uzeta_t^\tau\ozeta_t^\tau+1)+C_2\\
&\leq&C_1S_0\exp\{\sum_{t=1}^\tau\uzeta_t^\tau\ozeta_t^\tau\}+C_2\\
&<&C_3
\end{eqnarray*}}
for some constant $C_3>0$, by our assumption that $\sum_{t=1}^\tau\uzeta_t^\tau\ozeta_t^\tau\to\nu$.
\end{proof}





\subsection{Concave payoffs and lower bounds}\label{asec:concavelower}

\myparab{Concave payoffs.} We have a simple characterization of the hedging game's equilibrium when the payoff function is concave, under general conditions on $\mathcal U_t$:
\begin{corollary}
Consider a $\tau$-round game. When the payoff function $g(\cdot)$ is concave, with uncertainty sets $\{\mathcal U_t\}_{t\in[\tau]}$ each of which contains the point 0, the option price's upper bound is
\begin{eqnarray}
&&\min_{\Delta_t, t \in [\tau]}\max_{R_t \in \mathcal U_t, t \in [\tau]}g\left(S_0\prod_{t = 1}^{\tau}(1+R_t)\right)- \sum_{t = 1}^{\tau}(R_t\Delta_t)=g(S_0)
\end{eqnarray}
\end{corollary}

\begin{proof}
Consider a single-round game \ie $\tau=1$. Recall Proposition~\ref{prop:dual}, which states that the upper bound is $\max_{P_f\in\mathcal{P}(\mathcal U):\E_{P_f}[R]=0}E_{P_f}[g(S_0(1+R))]$. By Jensen's inequality, $\E_{P_f}[g(S_0(1+R))]\leq g(S_0(1+\E_{P_f}[R]))=g(S_0)$ for any $P_f\in\mathcal{P}(\mathcal U)$ such that $\E_{P_f}[R]=0$. The result is then immediate for $\tau=1$.

The conclusion from multi-round game follows exactly the same as the argument for Lemma \ref{lem:discrete} (now concavity is preserved in every step in the backward induction).
\end{proof}

\myparab{Lower bounds.} By replacing $g$ with $-g$ in all analysis above, we immediately get results for lower bounds. The following is analogous to Proposition \ref{prop:dual}:
\begin{proposition} Let $\tau=1$ and $S_0$ be the initial price. Consider a bounded uncertainty set $\mathcal U$ and a continuous payoff function $g$. The lower bound of the option price is
\begin{equation}\label{eqn:charl}
\max_{\Delta}\min_{R \in \mathcal U}g(S_0(1+R)) - R \Delta =  \min_{\substack{P_f \in \mathcal P(\mathcal U),\\ \E_{R\leftarrow P_f}[R] = 0}}\E_{P_f}[g(S_0(1+R))],
\end{equation}
where $\mathcal{P}(\mathcal{U})$ denotes the set of all probability measures $P_f$ that have support on $\mathcal{U}$. The maximization problem in the right hand side above is over all such probability measures that satisfy $\E_{P_f}[R]=0$. \label{prop:dual lower} \end{proposition}

The following summarizes the characterizations for convex and concave payoffs:
\begin{corollary}
Consider the $\tau$-round European option hedging game, with uncertainty sets $\{\mathcal U_t\}_{t\in[\tau]}$. The following results hold:
\begin{enumerate}
\item Suppose the payoff function $g$ is concave. If the uncertainty sets $\mathcal U_t=[-\uzeta_t,\ozeta_t]$ for all $t\in[\tau]$, then the lower bound is $\E_{P_f}[g(S_0\prod_{t=1}^\tau(1+R_t))]$, where $P_f$ is the unique risk-neutral measure supported on $\{-\uzeta_t,\ozeta_t\}$ for each $R_t$, i.e. $E_{P_f}[R_t]=0$.
\item Suppose the payoff function $g$ is convex, and the uncertainty sets $\mathcal U_t$ all contain the point 0. Then the lower bound is $g(S_0)$.
\end{enumerate}\label{prop:lowerdual}
\end{corollary}

\section{Missing analysis for non-convex payoffs}
\subsection{Proof of Lemma~\ref{lem:lipschitz}}\label{asec:lipschitz}
We prove by (backward) induction. Obviously $g_{\tau}(x)=g(x)$ satisfies the Lipschitz condition. Suppose $g_{t+1}(x)$ satisfies $g_{t+1}(x)-g_{t+1}(y)\leq L(x-y)$ for $x\geq y$. We prove $g_{t}(x)-g_{t}(y)\leq L(x-y)$ for $x\geq y$ by contradiction.

Assuming this is not true, then
$$\max_{P\in\mathcal{P}(\mathcal U):\E[R]=0}\E[g_{t+1}((x+\delta)(1+R))]>\max_{P\in\mathcal{P}(\mathcal U):\E [R]=0}\E[g_{t+1}(x(1+R))]+L\delta$$
for some $x\in \mathcal U$ and $\delta>0$. Now let $\tilde{P}_f$ be an optimal solution for $\E[g_{t+1}((x+\delta)(1+R))]$ (the existence of an optimal solution will be seen immediately in the next lemma). Then
\begin{eqnarray*}
\E_{\tilde P_f}\left[g_{t+1}((x+\delta)(1+R))\right]&=&\max_{P_f\in\mathcal{P}(\mathcal U):\E[R]=0}\E[g_{t+1}((x+\delta)(1+R))]\\
& > & \max_{P_f\in\mathcal{P}(\mathcal U):\E[R]=0}\E[g_{t+1}(x(1+R))]+L\delta\\
&\geq&\E_{\tilde P_f}[g_{t+1}(x(1+R))] + L\delta.
\end{eqnarray*}
But since $g_{t+1}(x)$ is assumed to be Lipschitz continuous, we have
$$\E_{\tilde P_f}[g_{t+1}((x+\delta)(1+R))]-\E_{\tilde P_f}[g_{t+1}(x(1+R))]\leq \E_{\tilde P_f}[L\delta(1+R)]=L\delta$$
by the risk-neutral property of $\tilde{P}_f$. This leads to a contradiction.

\subsection{Proof of Lemma~\ref{lem:gtgm}}\label{asec:lem:gtgm}
We shall first show that the primal formulation
\begin{equation}\label{eqn:localprimal}
\begin{array}{ll}
\mbox{minimize} &p\\
\text{subject to}&p+r\Delta\geq g_{t+1}(x(1+r))\text{\ for all $r\in \mathcal U$}
\end{array}
\end{equation}
has the following property in terms of \emph{binding constraints}:
%
%
%
%
\begin{lemma}\label{cor:twoworld}
For any $t$ and $x$, there exists an optimal solution for \eqref{eqn:localprimal}, say $(p^*,\Delta^*)$, such that either:
\begin{enumerate}
\item $p^* =g_{t+1}(x)$ or
\item There are exactly two binding constraints, corresponding to $r^{(1)}$ and $r^{(2)}$, such that
$r^{(1)} > 0$ and $r^{(2)} < 0$. These two constraints uniquely define $(p^*,\Delta^*)$.
\end{enumerate}
\end{lemma}

\begin{proof}[Proof of Lemma~\ref{cor:twoworld}]
First, there must be at least one binding constraint for \eqref{eqn:localprimal}, because if not, one can always decrease $p$ to achieve lower objective value while preserving all constraints. On the other hand, there must be at most two binding constraints for \eqref{eqn:localprimal}; otherwise, there will be three linearly independent equations $p^*+r_i \Delta^*  =  g_{t + 1}(x(1+r_i))$ for some $r_i,i=1,2,3$ that $(p^*,\Delta^*)$ satisfies, which is impossible.

Next, let us move to the case where there are two binding constraints. We need to show that $\rone < 0 < \rtwo$. Suppose $\rone, \rtwo > 0$. We can see that the dual program of (\ref{eqn:localprimal}) is infeasible, which means the primal LP is either infeasible or unbounded. Infeasibility is easily ruled out since one can put $\Delta=0$ and a large enough $p$ to construct a feasible solution. To show that unboundedness is also impossible, we will show that all $p$ that are smaller than a negative threshold are infeasible, implying the minimization problem \eqref{eqn:localprimal} is finite. To this end, consider any $\tilde{r}_1>0$ and $-\tilde{r}_2<0$ that lie in $\mathcal U$. Suppose $p+\tilde{r}_1\Delta\geq g_{t+1}(x(1+\tilde{r}_1))$. This implies $\Delta\geq(g_{t+1}(x(1+\tilde{r}_1))-p)/\tilde{r}_1$. Now, if we choose $p$ to be very negative, then
$$p-\tilde{r}_2\Delta\leq p-\tilde{r}_2\frac{g_{t+1}(x(1+\tilde{r}_1))-p}{\tilde{r}_1}=\left(1+\frac{\tilde{r}_1}{\tilde{r}_2}\right)p-\frac{\tilde{r}_2}{\tilde{r}_1}g_{t+1}(x(1+\tilde{r}_1))<g_{t+1}(x(1-\tilde{r}_2))$$
and hence $(p,\Delta)$ does not satisfy $p-\tilde{r}_2\Delta\geq g_{t+1}(x(1-\tilde{r}_2))$. Hence the set of feasible $p$ must be bounded from below. Similarly, we can show that it is impossible to have $\rone, \rtwo < 0$. It is trivial to see that the two binding constraints uniquely define $(p^*,\Delta^*)$.

Finally, suppose there is exactly one binding constraint. When $r$ in the corresponding binding constraint is non-zero, the dual LP is again infeasible, which will result in a contradiction again as above.
\end{proof}

We may now proceed to prove Lemma~\ref{lem:gtgm}. Recall that we use the following LP to compute $g_t(x)$:

\begin{equation}\label{eqn:exact}
\begin{array}{ll}
\mbox{minimize}&p\\
\text{subject to}&g_{t + 1}(x(1+r))-r\Delta\leq p \text{\ \ for $r \in \mathcal U$}
\end{array}
\end{equation}

Let $(p^*, \Delta^*)$ be the optimal solution.
By Lemma~\ref{cor:twoworld}, we have either 1) $p^* = g_{t+1}(x)$ or 2) there exists a pair $\rone$ and $\rtwo$ such that $\rone < 0 < \rtwo$ and
$$p^* + r^{(i)}\Delta^* = g_{t+1}(x(1+r^{(i)}) \mbox{ for } i = 1, 2.$$
Let us focus on the second case, as the first case follows similarly. In this case, we can express $p^*=g_t(x)$ in terms of $\rone$ and $\rtwo$ as follows:
\begin{equation}\label{eqn:g_t simplified}
g_t(x) = \frac{\rtwo}{\rtwo - \rone}g_{t+1}(x(1+\rone)) - \frac{\rone}{\rtwo - \rone}g_{t+1}(x(1+\rtwo))
\end{equation}

Since the discretization length is $\epsilon$, there exist $\hrone < 0 < \hrtwo \in \hmu$ such that $0\leq\rone - \hrone < \epsilon$ and $0\leq\hrtwo - \rtwo < \epsilon$, \ie we can define $\hrone$ as the closest $r_i$ that is at least as large as $\rone$ and $\hrtwo$ as the closest $r_i$ that is at least as small as $\rtwo$. Note that a lower bound of $g^m_t(x)$ is given by the optimal solution of the following linear program:

\begin{equation}\label{eqn:exact}
\begin{array}{ll}
\mbox{minimize}&p\\
\text{subject to}&g_{t + 1}(x(1+\hrone))-x\hrone\Delta\leq p\\
& g_{t + 1}(x(1+\hrtwo))-x\hrtwo\Delta\leq p
\end{array}
\end{equation}
This LP can be solved analytically, which gives us
\begin{equation}\label{eqn:gmupper}
g^m_t(x) \geq \frac{\hrtwo}{\hrtwo - \hrone}g_{t+1}(x(1+\hrone)) - \frac{\hrone}{\hrtwo - \hrone}g_{t+1}(x(1+\hrtwo)).
\end{equation}
Now let $\eta>\epsilon$ be a parameter to be decided later. We consider two cases.

\noindent{\em Case 1.} $\rtwo - \rone \geq \eta$. First we can see that 
\begin{equation}\label{eqn:prob1}
\frac{\hrtwo}{\hrtwo-\hrone}\geq\frac{\rtwo}{\rtwo-\rone}\left(1-\frac{\epsilon}{\eta}\right)
\end{equation}
since
\begin{align}
\frac{\hrtwo}{\hrtwo-\hrone}&=\frac{\rtwo}{\rtwo-\rone}+\frac{\hrone\rtwo-\hrtwo\rone}{(\hrtwo-\hrone)(\rtwo-\rone)} \notag\\
&=\frac{\rtwo}{\rtwo-\rone}+\frac{\rtwo(\hrone-\rone)-\rone(\hrtwo-\rtwo)}{(\hrtwo-\hrone)(\rtwo-\rone)} \notag\\
&\geq\frac{\rtwo}{\rtwo-\rone}-\frac{\rtwo\epsilon}{(\hrtwo-\hrone)(\rtwo-\rone)}\text{\ \ (since $-\hrone(\hrtwo-\rtwo)\geq0$ and $|\hrone-\rone|<\epsilon$)} \notag\\
&\geq\frac{\rtwo}{\rtwo-\rone}-\frac{\rtwo\epsilon}{(\rtwo-\rone)^2}\text{\ \ (by construction we have $\hrtwo-\hrone\geq\rtwo-\rone$)} \notag\\
&\geq\frac{\rtwo}{\rtwo-\rone}\left(1-\frac{\epsilon}{\eta}\right). \notag
\end{align}
Similarly, we have
\begin{equation}
-\frac{\hrone}{\hrtwo-\hrone}\geq-\frac{\rone}{\rtwo-\rone}\left(1-\frac{\epsilon}{\eta}\right).
\label{eqn:prob2}
\end{equation}
Hence we may continue (\ref{eqn:gmupper}) and get:
{\small
\begin{eqnarray*}
g_t^m(x)&\geq& \frac{\hrtwo}{\hrtwo - \hrone}g_{t + 1}(x(1+\hrone)) - \frac{\hrone}{\hrtwo - \hrone}g_{t+1}(x(1+\hrtwo)) \\
& \geq & \frac{\rtwo}{\rtwo-\rone}\left(1-\frac{\epsilon}{\eta}\right)\left(g_{t+1}(x(1+\rone))-xL(\rone-\hrone)\right) \\
& & \quad \quad - \frac{\rone}{\rtwo - \rone}\left(1-\frac{\epsilon}{\eta}\right)g_{t+1}(x(1+\rtwo))  \\
& & \quad \quad \quad \mbox{(by Lipschitz continuity of $g_{t+1}$, and Eq. \eqref{eqn:prob1} and \eqref{eqn:prob2})} \\
& = & \left(1-\frac{\epsilon}{\eta}\right)\left\{\frac{\rtwo}{\rtwo - \rone}g_{t+1}(x(1+\rone)) - \frac{\rone}{\rtwo - \rone}g_{t+1}(x(1+\rtwo))\right\} - xL\epsilon\left(1-\frac{\epsilon}{\eta}\right) \\
& = & g_t(x)(1-\frac{\epsilon}{\eta}) - xL \epsilon\text{\ \ (by Eq. \eqref{eqn:g_t simplified}).}
\end{eqnarray*}}

\noindent{\em Case 2.} When $\rtwo - \rone < \eta$. This implies $\rtwo < \eta$ and $\rone > -\eta$. We shall first show that $|g_{t + 1}(x) - g_t(x)|$ is small. Specifically,
{\small
\begin{eqnarray*}
& & |g_{t + 1}(x) - g_t(x)| \\
& = & \left|g_{t+1}(x) - \left(\frac{\rtwo}{\rtwo - \rone}g_{t+1}(x(1+\rone)) - \frac{\rone}{\rtwo -\rone}g_{t+1}(x(1+\rtwo))\right)\right| \\
& \leq & \frac{\rtwo}{\rtwo - \rone}\left|g_{t+1}(x) - g_{t+1}(x(1+\rone))\right| - \frac{\rone}{\rtwo - \rone}\left|g_{t+1}(x) - g_{t+1}(x(1+\rtwo))\right| \\
& \leq & 2Lx\frac{|\rone\cdot\rtwo|}{\rtwo - \rone} \quad {\mbox{(by Lipschitz continuity of $g_{t+1}(\cdot)$)}} \\
& \leq & 2Lx |\rone| \quad {\mbox{(since $\frac{\rtwo}{\rtwo-\rone}\leq1$)}}\\
&\leq&2Lx\eta
\end{eqnarray*}}
We next use the above inequality to compute a lower bound for $g^m_t(x)$:
{\small
\begin{eqnarray*}
& & g^m_t(x) \\
& \geq & \frac{\hrtwo}{\hrtwo - \hrone}g_{t+1}(x(1+\hrone)) - \frac{\hrone}{\hrtwo - \hrone}g_{t+1}(x(1+\hrtwo)) \quad \mbox{(by Eq.(\ref{eqn:gmupper}))} \\
& \geq & \frac{\hrtwo - \hrone}{\hrtwo - \hrone}g_{t+1}(x) - \frac{2Lx|\hrone \hrtwo|}{\hrtwo - \hrone} \quad \mbox{(Lipschitz condition)}\\
& \geq & g_{t+1}(x) - 2Lx|\hrone|  \quad \mbox{(again use the fact that $\frac{ \hrtwo}{\hrtwo - \hrone} \leq 1$)} \\
& \geq & g_t(x) - (2\eta + 2\epsilon)Lx.\quad \mbox{(since $|\hrone|\leq|\hrone-\rone|+|\rone|\leq\epsilon+\eta$)}
\end{eqnarray*}
}
Summarizing both cases, we have
\begin{eqnarray*}
g_t(x) - g^m_t(x) & \leq & \min_{\eta>\epsilon}\max\left\{g_t(x)\frac{\epsilon}{\eta} + xL \epsilon , (2\eta + 2\epsilon)Lx\right\}  \\
& \leq & \min_{\eta>\epsilon}\max\left\{g_t(x)\frac{\epsilon}{\eta}  , 2Lx\eta \right\}  + 2\epsilon Lx \\
& \leq & \sqrt{g_t(x)2\epsilon Lx} + 2L \epsilon x
\end{eqnarray*}
by setting $\eta = \sqrt{g_t(x)\epsilon/(2Lx)}$ when $\epsilon\leq\sqrt{g_t(x)/(2Lx)}$, and $\eta=\epsilon$ when $\epsilon>\sqrt{g_t(x)/(2Lx)}$.

\section{Missing proofs for American options}\label{sec:american}
\subsection{Proof of Lemma~\ref{lemma:american binomial}}\label{proof american binomial}
The proof follows similarly as that for European-type options. We will start by arguing that $g_t(x)$ is convex for any $t$ and
\begin{equation}
g_{t}(x) = \max\left\{\max_{\substack{P \in \mathcal P(\{-\uzeta, \ozeta\})\\ \E[R] = 0}}\E[g_{t+1}(x(1+R_{t+1}))], g(x)\right\}.\label{step1}
\end{equation}
by induction. First, as in Corollary~\ref{cor:reduction}, when $g_{t+1}(\cdot)$ is convex, the optimal risk-neutral measure for the dual problem at the $t$-th round, assuming the option is not exercised, has probability masses only at $\{-\uzeta, \ozeta\}$. Second, when $g_{t+1}(\cdot)$ is convex, $\max_{\substack{P \in \mathcal P(\{-\uzeta, \ozeta\})\\ \E[R] = 0}}\E[g_{t+1}(S(1+R))]$ is also convex by Lemma~\ref{lem:convex2convex}. Since the maximum of convex functions is still convex, we conclude the induction.

Note that \eqref{step1} can be written as
$$g_{t}(x) = \max\left\{\E_{P_f}[g_{t+1}(x(1+R_{t+1}))], g(x)\right\}.$$
where $P_f\in \mathcal P(\{-\uzeta, \ozeta\})$ is a risk-neutral measure satisfying $\E[R] = 0$. We argue that $g_t(x)\geq g(x)$ and $\E_{P_f}[g_{t}(x(1+R_{t}))]\geq g(x)$ for any $t$ and $x$, by induction again. Indeed, if $g_{t+1}(x)\geq g(x)$ for any $x$, then $\E_{P_f}[g_{t+1}(x(1+R_{t+1}))]\geq g_{t+1}(E_{P_f}[x(1+R_{t+1})])$ by Jensen inequality, which is equal to $g_{t+1}(x)$ since $E_{P_f}[R_{t+1}]=0$, and hence dominates $g(x)$ by the induction hypothesis. We have therefore shown the claim, which implies that it is optimal to exercise at the maturity.
%
%

\subsection{Continuous-time exercise right}\label{sec:continuous exercise}
One arguably unsatisfying feature in our model for the upper bound of American options is that the nature can only exercise the option at discrete rounds, \ie at the times when the trader can execute trade decisions. There is a natural formulation to relax this constraint, and it turns out that this relaxation does not change our existing model.

\begin{definition}[American option hedging game with continuous-time exercise right.]\label{def:strongadv} We model the dynamics for an American option, with payoff $g(\cdot)$ and expiration $T$, by considering a $\tau$-round game between the trader and the nature. The time length for each round is $\gamma \triangleq T/\tau$. Let $\{-\uzeta_i, \ozeta_i\}$ be the uncertainty parameters for the $i$-th round. In this game,
\begin{itemize}
\item the investor is only allowed to trade at the beginning of each round, \ie at time $0, \gamma, 2\gamma, ..., T$.
\item the adversary is allowed to exercise the option at any time. The adversary also decides the (continuous) trajectory of the price movement subject to the following constraints specified by the uncertainty parameters: let $t = i \gamma + \delta$, where $\delta < \gamma$; we require $-\uzeta_{i + 1}\cdot \delta/\gamma \leq S_{t}/S_{i\gamma} - 1 \leq \ozeta_{i + 1} \cdot \delta/\gamma$ for any $\delta<\gamma$.
\end{itemize}
\end{definition}
Intuitively, the price movement lies in the uncertainty set in the form of a ``cone'' that extends from time $i\gamma$ to $(i+1)\gamma$.
By using a simple change of variable trick, we have the following observation:

\begin{corollary} In the model in Definition \ref{def:strongadv}, the upper bound of an American option's price is the same as the upper bound from the ordinary hedging game introduced in Section~\ref{sec:genshort}. 
\end{corollary}

\section{Convergence to continuous-time control problems for non-convex options}\label{sec:control}
In this section, we show that our minimax upper bound of the option price, with possibly non-convex payoff, converges to the price based on a controlled diffusion process when the uncertainty sets are appropriately scaled.
To facilitate discussion, for this section we let the time steps be $0,\delta,2\delta,\ldots,T$ (for simplicity let $T$ be a multiple of $\delta$), and thus $\tau = \frac{T}{\delta}$. Let $\mathcal{U}^\delta=[-\uzeta^\delta,\ozeta^\delta]$, where $-\uzeta^\delta\triangleq-\uzeta\sqrt{\delta}$ and $\ozeta^{\delta} \triangleq \ozeta\sqrt{\delta}$. Let us recall that we write $S(t)$ as the continuous-time price of the underlying asset at time $t$, with $S(0) = S_0$ (this process is decided by nature).
We have the following theorem:

\begin{theorem}\label{thm:continuous}
Let $g$ be a Lipschitz continuous payoff function. As $\delta\to0$, the upper bound of the option price defined by
\begin{equation}\label{V continuous}
\min_{\Delta_1,\Delta_2,\ldots,\Delta_{\tau}}\max_{R_1^\delta,\ldots,R_{\tau}^\delta\in\mathcal{U}^\delta}g\left(S_0\prod_{i=1}^{\tau}(1+R_i^\delta)\right)-\sum_{i=1}^{\tau}\Delta_iR_i^\delta
\end{equation}
converges to 
\begin{equation}
G(0,S_0):=\max_\xi \E[g(S(T))]. \label{value continuous}
\end{equation}
Here $S(t)$ follows the dynamic
\begin{equation}
dS(t)=\sigma(u)S(t)dw(t) \label{eqn:scontinuous}
\end{equation}
where $w(t)$ is a standard Brownian motion, $\sigma(u)=u$ is the controlled volatility, and $\xi=\{u(t),0\leq t\leq T:u(t)\in U\}$ is an adapted control sequence. The domain of control $U=[0,\underline{\zeta}\overline{\zeta}]$.
\end{theorem}
Intuitively, this theorem asserts that the continuous-time limit of the option price's upper bound for non-convex payoff is Gaussian in nature, similar to the case of convex payoff and Black-Scholes model. The crucial difference with those cases, however, is that nature now has the additional power to choose (to reduce) the volatility at any point of time, and this is only helpful to nature in the case of non-convex payoff. Since we know from Section~\ref{sec:non-convex} that in the discrete-time setting with non-convex payoff, the nature does not necessarily choose the extremal point in the uncertainty set, it is not surprising that there is reduction in volatility at certain points of time in the continuous-time counterpart.
In Appendix~\ref{sec:bsconverge} we will also demonstrate how Theorem~\ref{thm:continuous} can be easily reduced to Black-Scholes model in the case of convex payoff.

The way to show Theorem~\ref{thm:continuous} is to view the diffusion process \eqref{eqn:scontinuous} as the starting object, and argue that the discrete-time game is an \emph{approximating Markov chain} of \eqref{eqn:scontinuous}~\cite{kushner2001numerical}. The following is the key to prove Theorem~\ref{thm:continuous}:

\begin{theorem}
Consider $x(t)=x_0+\int_0^t\sigma(w(s),u(s))dw(s)$ where $w(t)$ is a standard Brownian motion, and $u(t)$ is an adapted control sequence on a compact set $U$. Let $V(t,x)=\max \E_{(t,x)}[h(x(T))]$, where $h$ is continuous and $\E_{(t,x)}$ denotes the expectation conditional on $x(t)=x$.

Consider a Markov chain approximation as follows. Divide the time into steps of size $\delta$. Define a Markov chain $\{X^\delta(t)\}_{t=0,\delta,2\delta,\ldots,T}$ with transition $p^\delta(x,y|\alpha)$. Let $V^\delta(T,x)=h(x)$. For each step backward, solve
$$V^\delta(t,x)=\max_{\alpha\in U}\E_{(t,X^\delta(t)=x)}^{\delta,\alpha}[V^\delta(t+\delta,X^\delta(t+\delta))]$$
where $\E_{(t,X^\delta(t)=x)}^{\delta,\alpha}$ is the expectation taken using the transition $p^\delta(x,y|\alpha)$ for $X^\delta(t+\delta)$. Then we interpolate $X^\delta(\cdot)$ and $V^\delta(\cdot,x)$ such that they are piecewise constant on $t\in[0,T]$, \ie $X^\delta(s)=X^\delta(i\delta)$ and $V^\delta(s,x)=V^\delta(i\delta,x)$ for $i\delta\leq s<(i+1)\delta$.

Suppose that
\begin{enumerate}
\item $\sigma(x,\alpha)$ is Lipschitz continuous in $x$, uniformly in $U$.
\item The transition is locally consistent, \ie $\E_{(t,X^\delta(t)=x)}^{\delta,\alpha}[X^\delta(t+\delta)-X^\delta(t)]=o(\delta)$, $\E_{(t,X^\delta(t)=x)}^{\delta,\alpha}(X^\delta(t+\delta)-X^\delta(t))^2=\sigma(x,\alpha)^2\delta+o(\delta)$ for any $\alpha\in U$, and $\sup_{0\leq t\leq T}|X^\delta(t+\delta)-X^\delta(t)|\stackrel{a.s.}{\to}0$ as $\delta\to0$.
\item $h$ is uniformly integrable, \ie $\lim_{\eta\to\infty}\sup_{\delta}\E_{(t,x)}^{\delta,\alpha}[|h(X^\delta(T))|;|h(X^\delta(T))|>\eta]<\infty$ for any $t,x$.
\end{enumerate}
Then
$V^\delta(t,x)\to V(t,x)$ pointwise on $t\in[0,T],x\in\mathbb{R}^+$. \label{control theorem}
\end{theorem}

\begin{proof}[Outline of proof]
The proof is adapted from that of~\cite{kushner2001numerical}, p. 356, Theorem 1.4. As such we only provide an outline here. The proof calls on machinery in weak convergence analysis. Since there is no regularity condition on the control $u(t)$ in \eqref{eqn:scontinuous}, the first step is to enlarge the space of the process to measure-valued space. Namely, for any given control sequence $\{u(t)\}_{t\in[0,T]}$, we can write
\begin{equation}
x(t)=S_0+\int_0^t\int_U\sigma(\alpha)x(v)M(d\alpha dv) \label{SDE}
\end{equation}
where $M(d\alpha dv)$ is a real measure-valued continuous random process defined with, roughly speaking, the following properties: For any fixed measurable subset $A$ in $U$, $M(A,t)$ is a martingale that has quadratic variation $m(A,t)$ (so-called martingale measure; see \cite{kushner2001numerical}, p. 352, (1.9)). The quantity $m(\cdot,\cdot)$ is the so-called relaxed control, and is a measure that puts a delta mass if the control $u(t)\in A$ at time $t$ (\cite{kushner2001numerical}, p. 263, (5.1)). Under Lipschiz continuity of $\sigma(\cdot,\cdot)$ and compact control set $U$, there is a unique weak sense solution to \eqref{SDE} (\cite{kushner2001numerical}, p. 353, discussion under A1.1).

Now, one can similarly define $x^\delta(t)=S_0+\int_0^t\int_U\sigma(\alpha)x^\delta(v)M^\delta(d\alpha dv)$ under the relaxed control $m^\delta$ that is discretized and interpolated through the time steps $0,\delta,2\delta,\ldots,T$, and $M^\delta(A,t)$ has quadratic variation $m^\delta(A,t)$. By p. 356, Theorem 1.3 in~\cite{kushner2001numerical}, under the conditions of Lipshitz continuity of $\sigma(\cdot,\cdot)$ and local consistency (\ie conditions 1) and 2) in our theorem), there exists a subsequence $(M^\delta,m^\delta)$ that converges weakly to $(M,m)$. By uniform integrability (\ie condition 3)), we then have $\liminf_\delta V^\delta(t,x)\geq V(t,x)$ for any $t,x$ \cite{amin1994convergence}. The final step is to argue that the inequality is in fact matched by some control. This consists of choosing a so-called $\epsilon$-optimal solution of the relaxed control by using a finite-dimensional Wiener process and a finite-valued and piecewise constant control, and arguing that this solution provides an $\epsilon$-approximation to the optimal value (\cite{kushner2001numerical}, p. 355, Theorem 13.1.2).

\end{proof}

\begin{proof}[Proof of Theorem \ref{thm:continuous}]
We will show that our hedging game is an approximating Markov chain to the control problem given in \eqref{value continuous}, and has the three listed properties in Theorem \ref{control theorem}. First, the $\sigma(x,\alpha)$ in Theorem \ref{control theorem} equals $x\alpha$ in \eqref{eqn:scontinuous}, where $\alpha$ is the control lying in the compact set $U=[0,\underline{\zeta}\overline{\zeta}]$.

Second, we check the local consistency property. Define $\{G^\delta(t,x)\}_{t=0,\delta,2\delta,\ldots,T}$ as the value function in each step of the backward induction in the discrete-time hedging game defined in \eqref{V continuous}. Note that by Proposition \ref{prop:dual}, $G^\delta(i\delta,S^\delta(i\delta))=\max_{P\in\mathcal{P}([-\underline{\zeta}^\delta,\overline{\zeta}^\delta])}\E[G^\delta((i+1)\delta,S^\delta(i\delta)(1+R_i^\delta))]$, where we define $S^\delta(i\delta)=S_0\prod_{j=1}^{i}(1+R_j^\delta)$ for $i\in[\tau]$, and piecewise constant interpolation of $S^\delta(t)$ for other $t$. In fact, we can rewrite this as
\begin{equation}
\begin{array}{ll}
\max_{p_1,p_2,r_1,r_2}&p_1G(t+\delta,S^\delta(t)(1+r_1))+p_2G(t+\delta,S^\delta(t)(1+r_2))\\
\text{subject to}&p_1r_1+p_2r_2=0\\
&p_1+p_2=0\\
&p_1,p_2\geq0,-\underline{\zeta}^\delta\leq r_1\leq r_2\leq\overline{\zeta}^\delta
\end{array} \label{discrete}
\end{equation}
where we encode the maximal probability distribution in \eqref{discrete} by $\{p_1,p_2,r_1,r_2\}$. Let $\alpha^2\delta=p_1r_1^2+p_2r_2^2$. Then the range of $\alpha$ is $[0,\underline{\zeta}^\delta\overline{\zeta}^\delta/\delta]=[0,\underline{\zeta}\overline{\zeta}]$. We then have $\E_{(t,S^\delta(t)=x)}^{\delta,\alpha}[S^\delta(t+\delta)-S^\delta(t)]=0$ and $\E_{(t,S^\delta(t)=x)}^{\delta,\alpha}(S^\delta(t+\delta)-S^\delta(t))^2=\E_{(t,S^\delta(t)=x)}^{\delta,\alpha}[(S^\delta(t))^2(R^\delta)^2]=(S^\delta(t))^2\alpha^2\delta$ (here $R^\delta$ denotes i.i.d copy of $R^\delta_i$ in \eqref{V continuous}). Moreover, $|S^\delta(t+\delta)-S^\delta(t)|\leq\max\{\underline{\zeta}^\delta,\overline{\zeta}^\delta\}=\max\{\underline{\zeta},\overline{\zeta}\}\cdot\delta$ and so $\sup_{0\leq t\leq T}|S^\delta(t+\delta)-S^\delta(t)|\stackrel{a.s.}{\to}0$. (One can see the policy defined in terms of $\alpha$ above is not unique; nevertheless it does not affect local consistency.)

Finally, we claim that $g$ is uniformly integrable. We use an argument similar to the proof of Corollary~\ref{cor:limit}, and conclude that $\sup_{\delta}\E_{(t,x)}^{\delta,\alpha}[g(S^\delta(T))^2]\leq C_1S_0\exp\{T\uzeta\ozeta\}+C_2<\infty$ for some constants $C_1,C_2>0$, by Lipschitz continuity of $g$ and that $\alpha\in[0,\uzeta\ozeta]$. This will imply that $g$ is uniformly integrable.
\end{proof}

\subsection{Reduction to Black-Scholes model from controlled diffusion process}\label{sec:bsconverge}
This section elaborates on the above analysis 
and provides an alternative proof for the continuous-time convergence of our price upper bound to the Black-Scholes model when the payoff function is convex. We shall derive a heuristic partial differential equation (PDE) that characterizes the solution for the controlled diffusion in Theorem~\ref{thm:continuous}. Then, under additional convexity assumption, we will demonstrate that our PDE is rigorously defined and coincides with the PDE for the Black-Scholes model.

To begin, let us write down a Hamilton-Jacobi-Bellman (HJB) equation informally using \eqref{value continuous} and \eqref{eqn:scontinuous}. Define $G(t, x)=\max_{\xi}\E_{(t,x)}[g(S(T))]$, where $\E_{(t,x)}$ denotes the expectation conditional on $S(t)=x$ and using the optimal control (and hence $G(0,S(0))$ is as defined in \eqref{value continuous}). Assuming for the moment that $G\in\mathcal{C}^2$, we can heuristically write
\begin{align*}
G(t,x)&=\max_{\xi}\E_{(t,x)}[G(t+\delta,S(t+\delta))]\\
&=\max_{\xi}\E_{(t,x)}[G(t,x)+G_t(t,x)\delta+G_x(t,x)xR^\delta+\frac{1}{2}G_{xx}x^2(R^\delta)^2+\cdots]\text{\ \ (by Taylor's series)}\\
&=G(t,x)+G_t(t,x)\delta+\frac{1}{2}G_{xx}x^2\max_{\xi}\E_{(t,x)}[(R^\delta)^2]+\cdots\text{\ \ (since $\E_{(t,x)}[R^\delta]=0)$}\\
\end{align*}
where $R^\delta=S(t+\delta)/S(t)-1$, $G_x(t, x) = \frac{\partial }{\partial x}G(t,x)$,  $G_{xx}(t, x) = \frac{\partial^2}{\partial x^2}G(t,x)$ \etc

Since $U = [0, \uzeta\ozeta]$ and $\E_{(t,x)}[(R^\delta)^2]=\sigma(u)^2\delta$, one can attempt to establish that $G(t, x)$ is the solution of the following PDE:
\begin{equation}
G_t(t,x)+\frac{1}{2}\max_{u\in U}\sigma(u)^2x^2G_{xx}(t,x)=0 \label{HJB}
\end{equation}
with the boundary condition $G(T,x)=g(x)$. In general, the solution of this PDE may not exist in the classical sense, and there is no guarantee to coincide with the optimal solution to the control problem in (\ref{value continuous})~\cite{flemingsoner}. The following theorem, nevertheless, presents a verification of the PDE's solution as the control problem's optimum, under a priori smoothness condition on $G$:
\begin{theorem}
Suppose $G^*\in\mathcal{C}^2 $ is a solution for (\ref{HJB}) (hence implying that the payoff function $g$ must be in $\mathcal C^2$), and that there exists an optimal $u^*(t,x)$ as the optimal solution to the $\max$ in \eqref{HJB}. Then $G^*(0, x)$ is the optimal solution to the control problem in (\ref{value continuous}), and $u^*(t,x)$ is the optimal control in (\ref{eqn:scontinuous}).
 \label{verification}
\end{theorem}

\begin{proof}
The proof follows by a standard application of Ito's lemma. For any $t<t'$ and $x$, we have
{\small
\begin{align*}
&\E_{(t,x)}^*[G^*(t',S(t'))]-G^*(t,x)\\
&=\E_{(t,x)}^*\left[\int_t^{t'}G^*_t(t,S(t))dt+\int_t^{t'}G^*_x(t,S(t))dS_t+\int_t^{t'}\frac{1}{2}G^*_{xx}(t,S(t))\sigma^2(u^*(t,S(t)))dt\right]\\
&=\E_{(t,x)}^*\left[\int_t^{t'}\left(G^*_t(t,S(t))+\frac{1}{2}G^*_{xx}(t,S(t))\sigma^2(u^*(t,S(t)))\right)dt\right]\text{\ \ (since $S(t)$ is a martingale)}\\
&=0
\end{align*}}
by \eqref{HJB}, where $\E_{(t,x)}^*$ denotes the expectation taken when the control is $u^*$. On the other hand, any control $u(t,x)$ must satisfy
\begin{equation}
G_t(t,x)+\frac{1}{2}\sigma(u(t,x))^2x^2G_{xx}(t,x)\leq0 \label{HJB1}
\end{equation}
by the definition of \eqref{HJB}. Hence the same argument leads to
\begin{align*}
& \E_{(t,x)}^u[G(t',S(t'))]-G(t,x)\\
&=\E_{(t,x)}^u\left[\int_t^{t'}G_t(t,S(t))dt+\int_t^{t'}G_x(t,S(t))dS(t)+\int_t^{t'}\frac{1}{2}G_{xx}(t,S(t))\sigma^2(u(t,S(t)))dt\right]\\
&\leq0
\end{align*}
where $\E_{(t,x)}^*$ denotes the expectation taken when the control is $u$. Hence $\E_{(t,x)}^*[G^*_t(t',S(t'))]=G^*_t(t,x)\geq \E_{(t,x)}^u[G(t',S(t'))]$ for any $u$. Take $t=0$ and $x=S_0$, we obtain the result.
\end{proof}

From Theorem \ref{verification}, we can obtain the Black-Scholes price when the payoff is convex. Observe that \eqref{HJB} can in fact be written as
\begin{equation}
G_t(t,x)+\frac{1}{2}\underline{\zeta}\overline{\zeta}x^2G_{xx}(t,x)I(G_{xx}(t,x)\geq0)=0 \label{HJB2}
\end{equation}
with boundary condition $G(T,x)=g(x)$. 
Suppose $G\in\mathcal{C}^2$ is convex in $x$ and so $G_{xx}(t,\cdot)\geq0$, then the equation \eqref{HJB2} reduces to
$$G_t(t,x)+\frac{1}{2}\underline{\zeta}\overline{\zeta}x^2G_{xx}(t,x)=0$$
which is the ordinary Black-Scholes PDE with zero risk-free rate, whose solution is known to be convex and lies in $\mathcal{C}^2$ when the payoff $g$ is convex and in $\mathcal C^2$~\cite{steele}. Hence by Theorem~\ref{verification} it is the solution to the control problem in \eqref{value continuous}, which is the limit of our hedging game model.

Lastly, suppose that $G_{xx}(t,\cdot)$ is concave, and so $G_{xx}(t,x)<0$. The equation \eqref{HJB2} then reduces to $G_t(t,x)=0$, which implies that the solution is $g(x)$, constant over $t\in[0,T]$. If the payoff $g$ is concave, we know again by Theorem~\ref{verification} that $g(x)$ is the optimal value of \eqref{value continuous}. 

\section{Integration of non-stochastic framework with price jumps}\label{sec:jump}
A natural extension to the Black-Scholes framework is to allow ``shocks'' in the movement of the stock price, typically modeled by the stochastic community as Poisson arrivals of jumps on top of a continuous geometric Brownian motion on the asset price. In this section we will adopt our adversary framework to incorporate price jumps.

There can be various ways to model when and how much the adversary can control the price to jump. Below we analyze two natural examples that can be extended from our framework in the previous sections. In the first example, we assume the adversary has no control over the occurrence and magnitude of jumps. We will show that the price upper bound is exactly the same as the price in the standard jump diffusion model, when the payoff function is convex. In the second example, the adversary can control when the jump happens, subject to a constraint on the total number of jumps.
The magnitude of the jumps can also be assumed to be controllable by the adversary, and we will see that only small modifications to the algorithm presented in Section~\ref{sec:non-convex} are needed.


\myparab{Example 1.  The random jump model.}
This model assumes a convex payoff function and that the nature performs a jump with a prefixed small probability $q$ at each step. If a jump occurs, the nature moves according to a return $Y$ that is random (whose distribution is given). Otherwise, with $1-q$ probability, the nature will have freedom to choose its path inside the uncertainty set $\mathcal{U}$. We also assume $q$ is sufficiently small (this will be specified precisely in the sequel). 

Consider first a single-round game, with initial price of the underlying asset $S_0$. The upper bound formulation is
\begin{equation}
\min_\Delta\max_R (1-q)[(g(S_0(1+R)))-\Delta R]+q\E_Y[g(S_0(1+Y))-\Delta Y], \label{jump1}
\end{equation}
where $R$ is the return that nature can choose, if a jump does not occur. The expectation $\E_Y$ is with respect to the jump magnitude variable $Y$.

The formulation \eqref{jump1} can be written as
\begin{equation}
\begin{array}{ll}
\text{minimize}&p\\
\text{subject to}&(1-q)[(g(x(1+r)))-\Delta r]+q\E_Y[g(x(1+Y))-\Delta Y]\leq p\text{\ \ for all $r\in\mathcal{U}$}
\end{array}
\end{equation}
or
\begin{equation}
\begin{array}{ll}
\text{minimize}&p\\
\text{subject to}&p+\Delta((1-q)r+q\E_Y[Y])\geq(1-q)g(x(1+r))+q\E_Y[g(x(1+Y))]\text{\ \ for all $r\in\mathcal{U}$}
\end{array} \label{primal jump1}
\end{equation}
The dual of \eqref{primal jump1} is
\begin{equation}
\begin{array}{ll}
\max&(1-q)\E[g(x(1+r))]+q\E_Y[g(x(1+Y))]\\
\text{subject to}&(1-q)\E[r]+q\E_Y[Y]=0\\
&P\in\mathcal{P}(\mathcal{U})
\end{array} \label{dual jump1}
\end{equation}
where $\mathcal{P}(\mathcal{U})$ denotes the set of all probability measures supported on $\mathcal{U}$.

We can use the same convexity argument as in Section~\ref{asec:convex} to argue that the optimal measure must be concentrated at the two extreme points of $\mathcal{U}$, namely $-\underline{\zeta}$ and $\overline{\zeta}$, when $q$ is small enough \ie $(1-q)(-\overline{\zeta})+q\E_Y[Y]>0$ and $(1-q)\underline{\zeta}+q\E_Y[Y]<0$. Let the weights of these two extreme points be $w_1$ and $w_2$. The constraints $(1-q)(w_1(-\underline{\zeta})+w_2\overline{\zeta})+q\E_Y[Y]=0$ and $w_1+w_2=1$ completely determine
\begin{equation}
w_1=\frac{\overline{\zeta}+\frac{q\E_Y[Y]}{1-q}}{\overline{\zeta}+\underline{\zeta}},\ w_2=\frac{\underline{\zeta}+\frac{-q\E_Y[Y]}{1-q}}{\overline{\zeta}+\underline{\zeta}} \label{optimal weights jump}
\end{equation}
Therefore, the upper bound is
\begin{equation}
(1-q)\left[\frac{\overline{\zeta}+\frac{q\E_Y[Y]}{1-q}}{\overline{\zeta}+\underline{\zeta}}g(x(1+\underline{\zeta}))+\frac{\frac{\underline{\zeta}-q\E_Y[Y]}{1-q}}{\overline{\zeta}-\underline{\zeta}}g(x(1+\overline{\zeta}))\right]+q\E_Y[g(x(1+Y))]. \label{jump optimal value}
\end{equation}
For a $\tau$-round game, the formulation follows analogously as in Section~\ref{sec:convexmultiround}, with the optimal value being the value function of a dynamic program, with $g_\tau(x)=g(x)$ and $g_t(x)$ equal to \eqref{jump optimal value} but with $g$ replaced by $g_{t+1}$. This characterization is exactly the same as a discrete jump diffusion process, studied in~\cite{amin1993jump}, which also demonstrated the continuous-time limit under appropriate scaling of the jump probability and the binomial return rates. This leads immediately to the following:
\begin{proposition}
Consider the random jump model with jump probability $q^\tau=q/\tau$ for a constant $q$, jump magnitude random variable $Y$ (not scaled with $\tau$), and uncertainty set $U^\tau=[-\uzeta/\sqrt{\tau},\ozeta/\sqrt{\tau}]$ at each step $t\in[\tau]$. Assume a convex payoff $g(\cdot)$. The upper bound of the option price converges to $\E[g(S(T))]$, where $S(t)$ follows a jump diffusion process given by
$$S(t)=S_0\exp\{\nu w(t)-\nu/2+\sum_{j=1}^{J(t)}\log Y(j)\}.$$
Here $\nu=\uzeta\ozeta$, $J(t)$ is a Poisson process with rate $p$, and $Y(j)$ are i.i.d. copies of $Y$.
\end{proposition}

\myparab{Example 2. The adversarial jump model.} In this model, we allow the adversary to make no more than $\ell$ jumps throughout the $\tau$ rounds of game. When the adversary decides to make a jump, it can choose a return from the uncertainty set $\mathcal W$; Otherwise, it can only choose from the ordinary uncertainty set $\mathcal U$. We assume the sizes of $\mathcal U$ and $\mathcal W$ are both polynomial in $\tau$. Typically $\mathcal W\supset\mathcal U$, but it is not required in the analysis.\footnote{For simplicity, we assume the uncertainty sets when jumps are present and absent are both uniform. This condition can easily be relaxed.}


We now explain how our analysis in Section~\ref{sec:btree} and our algorithm in Section~\ref{sec:non-convex} can be extended to this scenario. Define $g_t(x, \ell)$ as the price upper bound at the $t$-th round when the underlying asset's price is $x$ and the nature still has $\ell$ number of jump quota. At the $t$-th round, if $\ell\geq1$, the adversary may choose to use a jump, in which case the relevant price upper bound at the next round will be $g_{t+1}(x(1+R),\ell-1)$, where $R\in\mathcal W$; suppose the adversary chooses not to jump, then the relevant price upper bound becomes $g_{t+1}(x(1+R),\ell)$, where $R\in\mathcal U$. The LP formulation for each round is thus
\begin{equation}\label{eqn:jumpexact}
\begin{array}{ll}
\mbox{minimize}&p\\
\text{subject to}&g_{t + 1}(x(1+r), \ell )-r\Delta\leq p  \text{\ \ for $r\in\mathcal U$}\\
&g_{t + 1}(x(1+r), \ell - 1)-r\Delta\leq p  \text{\ \ for $r\in\mathcal W$, if $\ell \geq 1$.}
\end{array}
\end{equation}

%
%
By an argument similar to Proposition \ref{prop:dual}, we can introduce a risk-neutral measure that characterizes the optimal dual solution for \eqref{eqn:jumpexact}. The additional feature here is that the risk-neutral measure comprises of a mixture between the ordinary uncertainty set and the enlarged uncertainty set, depending on whether a jump is initiated. The dual formulation can be written as
$$\max_{\substack{q\in[0,1],P_f\in\mathcal P(\mathcal U),P_J\in\mathcal P(\mathcal W):\\q\E_f[R_{t+1}]+(1-q)\E_J[R_{t+1}]=0}}q\E_{P_f}[g_{t + 1}(x(1+R_{t+1}), \ell )]+(1-q)\E_{P_J}[g_{t + 1}(x(1+R_{t+1}), \ell - 1)]$$
Here $q$ is the mixture probability of the occurrence of jump, and $P_f$ and $P_J$ are the conditional distributions supported on $\mathcal U$ and $\mathcal W$ respectively, depending on whether a jump occurs. When $\ell=0$, then the dual formulation reduces to the case in Section~\ref{sec:non-convex}, and we merely have
$$\max_{P_f\in\mathcal P(\mathcal U)\E_f[R_{t+1}]}\E_{P_f}[g_{t + 1}(x(1+R_{t+1}), 0 )]$$
The dynamic program has the terminal value $g_\tau(x,\ell)=g(x)$ for all $\ell$ and $g_t(x, 0) = g_t(x)$ for all $t$, where $g_t(x)$ is defined in Section~\ref{sec:non-convex}.

%

From this characterization, we can use a multinomial approximation scheme similar to that in Section~\ref{sec:non-convex} to compute the price upper bound for general payoff functions. In this algorithm we discretize the uncertainty sets $\mathcal U$ and $\mathcal W$ into $\hmu$ and $\hmw$ so that the step length in each discrete set is $\epsilon$. The new feature is that in each backward induction step, we compute $\hat g_t(x,m)$ for $m\in\{0,1,\ldots,\ell\}$, where $x$ can take on polynomial number of distinct values (we can discretize $\mathcal U$ and $\mathcal W$ in a coherent way to achieve this).
Here, we pay a factor of $\ell$ in the running time and the space comlexity compared to the algorithm in Section~\ref{sec:non-convex} because of the exhaustive enumeration regarding $\ell$. We have the following analogous performance bound:
\begin{corollary}Let $\delta$ be an arbitrarily small constant. By setting $\epsilon=cL^2\tau^2/\delta^2$ for some constant $c$, our algorithm gives $\hat g_0(S_0)$ such that $g_0(S_0)-\delta S_0\leq\hat g_0(S_0)\leq g_0(S_0)$.
\end{corollary}

\begin{proof}
The proof can be adapted easily from that of Theorem~\ref{thm:main}, hence we shall highlight the main steps here. First, we extend, in a straightforward manner, the argument in Lemma~\ref{lem:lipschitz} to prove that $g_t(x,\ell)$ is Lipschitz continuous for any $t\in[\tau]$ and $\ell$. Next, the same binding constraint argument as in Lemma~\ref{lem:gtgm} will reveal that $0 \leq g_t(x,\ell) - g^m_t(x,\ell) \leq \sqrt{2L\max\{g_t(x,\ell),g_t(x,\ell-1)\} S\epsilon } + L x\epsilon$ for any $t$ and $\ell$, where for convenience we use the convention that $g_t(x,-1)=g_t(S,0)$. Then, since $g_t(x,\ell)\leq Lx$ for any $\ell$ by an argument similar to Lemma~\ref{lem:lipschitz}, we can use the ``artificial" probabilistic machinery to arrive at the conclusion.
\end{proof}


\section{Hardness results}\label{a:hard}
\subsection{Computational lower bound: $\#P$-hard result for non-uniform uncertainty sets}
This section proves that the problem of exactly computing the price upper bound is $\#P$-hard for non-uniform uncertainty sets, even when the payoff function is convex.

We reduce our problem from the counting subset sum problem. Let $A = \{a_1, ..., a_{\tau}\}$ be a set of positive integers. The counting problem here is to count the number of subsets $T \subseteq A$ such that the sum of integers in $T$ equals to $b$. This problem is known to be $\#P$-hard~\cite{CHW11}.

We use a Cook-Reduction (see \eg Chapter 17 in~\cite{arora2009computational}) and our reduction proceeds as follows. Given a counting subset sum instance $\{A, b\}$, we construct the price upper bound for a $\tau$-round European option hedging game such that the uncertainty set $\mathcal U_i$ at the $i$-th round is $[-\uzeta_i, \ozeta_i]$, where $\uzeta_i = \ozeta_i= \frac{e^{a_i} - 1}{e^{a_i}+1}$. The payoff function is the ordinary call option with strike $K$, \ie $g(x) = (x - K)^+$.
Since the payoff is convex, the optimal risk-neutral probability measure only has uniform probability mass on $\{-\uzeta_i, \ozeta_i\}$ at each round, \ie the price moves up by a factor $(1+\ozeta_i)$ with probability $\frac 1 2$ and moves down by a factor $(1-\uzeta_i)$ with probability $\frac 1 2$ for each round.

Let $P^*$ be this optimal probability measure. We next build a natural coupling between sampling a subset $T \subseteq A$ and moving a price trajectory in the binomial tree: an element $a_i \in T$ if and only if the price trajectory moves up at the $i$-th round. Under this probability measure, the total number of subsets $T$ in which the elements sum up to $b$ is
\begin{equation*}
2^{\tau}\Pr_{T \gets P^*}[\sum_{a_i \in T}a_i = b].
\end{equation*}
We next express the probability $\Pr_{T \gets P^*}[\sum_{a_i \in T}a_i = b]$ in terms of option prices. When $T$ is sampled from $P^*$, the coupled trajectory's final price is
\begin{equation*}
 S_0\prod_{a_i \in T}(1+\ozeta_i) \prod_{a_i \notin T}(1-\uzeta_i)
 =  S_0 \prod_{a_i \in A}\frac{2}{e^{a_i} + 1}\prod_{a_i \in T}e^{a_i}
 = S_0 \prod_{a_i \in A}\frac{2}{e^{a_i} + 1}\exp\left(\sum_{a_i \in T}a_i\right).
\end{equation*}

Let $\mathcal C = \prod_{a_i \in A}\frac{2}{e^{a_i} + 1}$. We can see from the above equation that
\begin{equation*}
\Pr_{T\gets P^*}[\sum_{a_i \in T}a_i = b] = \Pr_{P^*}[S_{\tau} = S_0\mathcal C\exp(b)].
\end{equation*}

Let us consider options with three different strike prices $K_1 = S_0\mathcal C \exp(b - 1)$, $K_2 = S_0\mathcal C\exp(b)$, and $K_3 = S_0\mathcal C\exp(b+1)$. Let their corresponding prices be $V_1$, $V_2$, and $V_3$. We can compute $V_1$ as follows:
\begin{equation*}
V_1 = \int (S_{\tau} - K)^+ dP^* = \sum_{j \geq b}\Pr[\sum_{a_i \in T}a_i = j](S_0\mathcal C \exp(j) - K_1).
\end{equation*}
Similarly, we have
\begin{equation*}
V_2 =  \sum_{j \geq b + 1}\Pr[\sum_{a_i \in T}a_i = j](S_0\mathcal C \exp(j) - K_2),
\end{equation*}
and
\begin{equation*}
V_3 =  \sum_{j \geq b + 2}\Pr[\sum_{a_i \in T}a_i = j](S_0\mathcal C \exp(j) - K_3).
\end{equation*}
From the above equalities, we have
\begin{equation*}
\begin{array}{ll}
V_1 - V_2 & = \Pr[\sum_{a_i \in T}a_i \geq b]S_0\mathcal C(\exp(b) - \exp(b-1)) \\
V_2 - V_3 & = \Pr[\sum_{a_i \in T}a_i \geq b + 1]S_0\mathcal C(\exp(b+1) - \exp(b)). \\
\end{array}
\end{equation*}
Therefore, we have
$$\Pr[\sum_{a_i \in T}a_i = b] = \frac{V_1 - V_2}{S_0 \mathcal C(\exp(b) - \exp(b-1))} -  \frac{V_2 - V_3}{S_0 \mathcal C(\exp(b+1) - \exp(b))},$$
which completes our reduction. 
\subsection{Information theoretic lower bound: additive dependencies on $S_0$}\label{sec:adddep}
We now present a lower bound under the oracle model to justify the necessity of our algorithm's additive dependency
on the stock's initial price $S_0$ and the Lipschitz parameter $L$, even for a single-round model. Specifically, we shall show that in a single-round hedging game, if the number of queries to the oracle is $b$, then there exists two payoff functions $g(x)$ and $h(x)$ such that:
\begin{enumerate}
\item Both functions are monotonic, $L$-Lipschitz, and have the same values at the queried points.
\item There exists an $S_0$ such that the difference between the price upper bounds at $S_0$ for the functions is $\Theta(LS_0/b)$.
\end{enumerate}
In other words, so long as the number of queries is only polynomial in $\tau$, there will be an additive error that is linear in $S_0$ and $L$.

We now explain our construction. Let $u$ be the size of the uncertainty set. Let the query points be $q_1 < q_2 < ... < q_i < 0 < ... < q_b$, where $q_b - q_1 \leq u$. For expositional purpose, let us assume there are at least two queried points on the negative axis. Also, let us assume $\mathcal U$ is symmetric, \ie $\mathcal U = [-\frac u 2 ,\frac u 2]$. Both assumptions can easily be relaxed.

By using an averaging argument, we see that there exists $i$ and $j$ such that $(q_i < 0 \mbox{ and } q_i - q_{i - 1} \geq \frac u {2b})$ and $(q_j > 0 \mbox{ and } q_{j + 1} - q_j \geq \frac u {2b})$. Now we define $g(x)$ and $h(x)$ as follows:

\mypara{Definition of $g(x)$:} Let $C$ be a sufficiently large number, \eg $C = (1+100 u) S_0$.
\begin{equation}
g(x) = \left\{\begin{array}{ll}
\frac{Lx}2 & x \leq C \\
Lx & x > C.
\end{array}
\right.
\end{equation}

\mypara{Definition of $h(x)$:} Let $I_1 = [(1+q_{i - 1})S_0, (1+q_i)S_0]$ and $I_2 = [(1+q_j)S_0, (1+q_{j + 1})S_0]$.
\begin{equation}
h(x) =
\left\{
\begin{array}{ll}
g(x) & x \notin I_1, I_2.\\
Lx & x \in [(1+q_{i - 1})S_0,(1+\frac{q_{i - 1}+q_i}2)S_0] \\
L(1+q_i)S_0/2& x \in [(1+\frac{q_{i - 1}+q_i}2)S_0, (1+q_i)S_0]\\
Lx  & x \in [(1+q_{j})S_0,(1+\frac{q_{j}+q_{j+1}}2)S_0]\\
L(1+q_{j+1})S_0/2 & x \in [(1+\frac{q_{j}+q_{j+1}}2)S_0, (1+q_{j + 1})S_0].
\end{array}
\right.
\end{equation}
We can see that the price upper bound for $g(S_0)$ in this single-round model is $LS_0/2$. For $h(x)$, we can see that in the dual characterization of its optimal solution, the corresponding risk-neutral measure has probability masses only at $r_1 \triangleq \frac{q_{i} + q_{i-1}}{2}$ and $r_2 \triangleq \frac{q_{j+1} + q_{j}}{2}$. Thus, the price upper bound for $h(S_0)$ is at least $S_0(\frac L 2 + \frac{uL}{2b})$, which completes our argument.

\end{document}